\documentclass[10pt,conference]{IEEEtran}

\renewcommand\thesection{\arabic{section}}
\renewcommand\thesubsection{\thesection.\arabic{subsection}}
\renewcommand\thesubsubsection{\thesubsection.\arabic{subsubsection}}

\renewcommand\thesectiondis{\arabic{section}}
\renewcommand\thesubsectiondis{\thesectiondis.\arabic{subsection}}
\renewcommand\thesubsubsectiondis{\thesubsectiondis.\arabic{subsubsection}}

\usepackage{verbatim} 
\usepackage{clrscode3e}
\usepackage[pdftex]{graphicx}
\usepackage{setspace}
\usepackage[cmex10]{amsmath}
\usepackage{amsmath, amsthm, amssymb}
\usepackage{algorithmic}
\usepackage{algorithm}
\usepackage{subfig}
\usepackage{caption}
\usepackage{paralist}
\usepackage{cases}
\usepackage{cite}
\usepackage{etoolbox}

\usepackage{tikz}

\usepackage{color}

\newcommand{\ceiling}[1]{\left\lceil{#1}\right\rceil}
\newcommand{\floor}[1]{\left\lfloor{#1}\right\rfloor}
\newcommand{\setof}[1]{\left\{{#1}\right\}}

\newcommand{\frameworkkq}[1]{$\mathbf{k^2Q}$}
\newcommand{\frameworkku}[1]{$\mathbf{k^2U}$}

 \def\myendproof{{\ \vbox{\hrule\hbox{%
   \vrule height1.3ex\hskip0.8ex\vrule}\hrule }}\par}
 \renewenvironment{proof}{\noindent{\bf Proof. }}{\myendproof}
 \newenvironment{appProof}[1]{\noindent{\bf Proof of
     #1. }}{\myendproof\vskip 0.1in}

 \setboolean{ALC@noend}{true}
\providebool{techreport}
\setbool{techreport}{true}

\newtheorem{theorem}{Theorem}
\newtheorem{lemma}{Lemma}
\newtheorem{corollary}{Corollary}
\newtheorem{example}{Example}

\newtheorem{definition}{Definition}

\graphicspath{{fig/multiframe/}{fig/arbitrary/}{fig/dag/}} 

\ifbool{techreport}{
\addtolength{\textheight}{14pt}
}{
\addtolength{\textheight}{6pt}
}
\newcommand{\citetechreport}[1]{\ifbool{techreport}{}{ in the report \cite{DBLP:journals/corr/abs-k2q}}}

\usetikzlibrary{%
  arrows,%
  shapes.misc,
  shapes.arrows,%
  chains,%
  matrix,%
  positioning,
  scopes,%
  decorations.pathmorphing,
  shadows%
}

\pgfdeclarelayer{background}
\pgfdeclarelayer{foreground}
\pgfsetlayers{background,main,foreground}
\tikzstyle{materia}=[draw, fill=white, text width=1.0em, text centered,
  minimum height=4em,drop shadow]
\tikzstyle{practica} = [materia, text width=18em, minimum width=8em,
align =left,
  minimum height=3em, rounded corners, drop shadow]
\tikzstyle{texto} = [above, text width=6em, text centered]
\tikzstyle{linepart} = [draw, thick, color=blue!50, -latex', dashed]
\tikzstyle{line} = [draw, line width = 2pt, color=blue!50, -latex']
\tikzstyle{ur}=[draw, text centered, minimum height=0.01em]

\newcommand{\practica}[2]{node (p#1) [practica]
  {\\{\footnotesize{#2}}}}

\ifbool{techreport}{
\pagestyle{plain}
}{
\pagestyle{empty}
}

\title{\huge \textbf{\textrm{k$^2$Q}}: A Quadratic-Form Response Time and Schedulability Analysis Framework for Utilization-Based Analysis}

\author{
    Jian-Jia Chen and Wen-Hung Huang\\
    Department of Informatics\\
    TU Dortmund University, Germany
    \and
    Cong Liu\\
    Department of Computer Science\\
    The University of Texas at Dallas
}
\vspace{3mm}

\begin{document}

\maketitle

\ifbool{techreport}{
\thispagestyle{plain}
}{
\thispagestyle{empty}
}

\begin{abstract}
  In this paper, we present a general response-time analysis and
  schedulability-test framework, called \frameworkkq{} (k to Q). It
  provides automatic constructions of closed-form quadratic bounds or
  utilization bounds for a wide range of applications in real-time
  systems under fixed-priority scheduling. The key of the framework is
  a $k$-point schedulability test or a $k$-point response time
  analysis that is based on the utilizations and the execution times
  of $k-1$ higher-priority tasks.  The natural condition of
  \frameworkkq{} is a quadratic form for testing the schedulability
  or analyzing the response time. The response time analysis and the
  schedulability analysis provided by the framework can be viewed as a
  ``blackbox'' interface that can result in sufficient
  utilization-based analysis.  Since the framework is independent from
  the task and platform models, it can be applied to a wide range of applications.
\ifbool{techreport}{

We show the generality of \frameworkkq{} by
  applying it to several different task models. \frameworkkq{} produces better uniprocessor and/or multiprocessor schedulability tests  not only for the traditional sporadic task model, but also
   more expressive task models such as the generalized multi-frame task model and the acyclic task model.  Another interesting contribution is that in the past, exponential-time schedulability tests were typically
  not recommended and most of time ignored due to high
  complexity. We have successfully shown that
  exponential-time schedulability tests may lead to good
  polynomial-time tests (almost automatically) by using the
  \frameworkkq{} framework.}{} \ifbool{techreport}{Analogously, a similar concept to test
  only $k$ points with a different formulation has been studied by us in
  another framework, called \frameworkku{}, which provides hyperbolic
  bounds or utilization bounds based on a different formulation of
  schedulability test.  With the quadratic and hyperbolic expressions,
  \frameworkkq{} and \frameworkku{} frameworks can be used to provide
  many quantitive features to be measured, like the total utilization
  bounds, speed-up factors, etc., not only for uniprocessor scheduling
  but also for multiprocessor scheduling. }{}
\end{abstract} 

\section{Introduction}
\label{sec:intro}

Analyzing the worst-case timing behaviour to ensure the timeliness of
embedded systems is essential for building reliable and dependable
components in cyber-physical systems. Due to the interaction and
integration with external and physical devices, many real-time and
embedded systems are expected to handle a large variety of
workloads. Towards such dynamics, several formal
real-time task models are established to represent these workloads with various
characteristics, such as the the generalized multi-frame task
model \cite{DBLP:journals/rts/BaruahCGM99,DBLP:conf/ecrts/StiggeY12} and the self-suspending task model~\cite{suspension}. 
\ifbool{techreport}{
 To analyze the worst-case response time or to ensure the timeliness of
the system, for each of these task models, researchers tend to develop
dedicated techniques that result in schedulability tests with
different computation complexity and accuracy of the
analysis. Although many successful results have been developed, after
many real-time systems researchers devoted themselves for many years,
there does not exist a general framework that can provide  efficient and
effective analyses for different task models.}{
Although many successful results have been developed, after
many real-time systems researchers devoted themselves for many years,
there does not exist a general framework that can provide efficient analyses for different task models.
 For analysis techniques developed over the years for analyzing different task models (exhibiting different characteristics though), redundancy may exist in them because many of the developed techniques tend to apply several fundamental analysis frameworks such as the response time analysis and the busy-window analysis  \cite{DBLP:conf/rtss/Lehoczky90}. Our motivation behind this paper is: can we design a general, powerful, yet easy-to-use schedulability analysis framework that is applicable to a wide selection of real-time task models?
}

\ifbool{techreport}{
Prior to this paper, we have presented a  general
schedulability analysis framework \cite{DBLP:journals/corr/abs-1501.07084,ChenHLRTSS2015}, called \frameworkku{}, that can be applied
in uniprocessor scheduling and multiprocessor scheduling, as long as
the schedulability condition can be written in a specific form to
test only $k$ points.  For example, to verify the schedulability of a
(constrained-deadline) sporadic real-time task $\tau_k$ under fixed-priority scheduling
in uniprocessor systems, the time-demand analysis (TDA) developed in
\cite{DBLP:conf/rtss/LehoczkySD89} can be adopted.

}{
Motivated by this, we present in this paper a rather general
schedulability analysis framework, called \frameworkkq{} (k to Q), that can be applied
in uniprocessor and multiprocessor scheduling for analyzing various real-time task models. 
}
The general concept to obtain sufficient schedulability tests in
the 
\frameworkkq{} framework is to test only a subset of time
points for verifying the schedulability. This idea is implemented in
the \frameworkkq{} framework by providing a $k$-point last-release
schedulability test, which only needs to test $k$ points under
\textit{any} fixed-priority scheduling when checking schedulability of
the task with the $k^{th}$ highest priority in the system.  Moreover,
this concept is further extended to provide a safe
upper bound of the worst-case response time. The response time
analysis and the schedulability analysis provided by the framework can
be viewed as a ``blackbox'' interface that can result in sufficient
utilization-based analysis. 

\noindent\textbf{Related Work.}
There have been several results in the literature with respect to utilization-based, e.g., \cite{liu1973scheduling,HanTyan-RTSS97,journals/tc/LeeSP04,DBLP:conf/rtas/WuLZ05,kuo2003efficient,bini2003rate,RTSS14a}, and non-utilization-based, e.g., \cite{DBLP:conf/rtss/ChakrabortyKT02,DBLP:conf/ecrts/FisherB05}, schedulability tests for the sporadic real-time task model and its generalizations in uniprocessor systems.
Most of the existing utilization-based schedulability analyses focus
on the total utilization bound. That is, if the total utilization of
the task system is no more than the derived bound, the task system is
schedulable by the scheduling policy. For example, the total
utilization bounds derived in
\cite{liu1973scheduling,HanTyan-RTSS97,DBLP:dblp_journals/tc/BurchardLOS95}
are mainly for rate-monotonic (RM) scheduling, in which the results in
\cite{HanTyan-RTSS97} can be extended for arbitrary fixed-priority
scheduling. Kuo et al. \cite{kuo2003efficient} further improve the
total utilization bound by using the notion of divisibility. Lee et
al. \cite{journals/tc/LeeSP04} use linear programming formulations for
calculating total utilization bounds when the period of a task can be selected.
Moreover, Wu et al. \cite{DBLP:conf/rtas/WuLZ05} adopt the
Network Calculus to analyze the total utilization bounds of several
real-time task models. 

Bini and Buttazzo \cite{DBLP:journals/tc/BiniB04} propose a framework
of schedulability tests that can be tuned to balance the time
complexity and the acceptance ratio of the schedulability test for
uniprocessor sporadic task systems.  The efficient tests in
\cite{DBLP:journals/tc/BiniB04} are based on an observation to test
whether the parameters of a task set fall into a schedulable region of
the fixed-priority scheduling policy.  Our strategy and philosophy are
simpler than \cite{DBLP:journals/tc/BiniB04}.  First, we only look at
the parameters of task $\tau_k$ (the task defined as the $k^{th}$
highest priority) that is under analysis by assuming that the
higher-priority tasks are already verified to be schedulable. Second,
similar to our recent general schedulability analysis framework
\frameworkku{} \cite{ChenHLRTSS2015}, we also apply the key idea of evaluating only $k$
points.  The tunable strategies in \cite{DBLP:journals/tc/BiniB04}
consider to examine a subset of the time points for schedulability
tests.

Distinct from the results in \cite{DBLP:journals/tc/BiniB04}, our
objective in this paper is to find closed-form schedulability tests
and response-time analyses that can be independent from task and
platform models. We target at sufficient schedulability tests and response time analyses that are not exact but can be calculated efficiently in linear-time or polynomial-time complexity.
\ifbool{techreport}{

\noindent{\bf Comparison to \frameworkku{}:} 
Even though \frameworkkq{} and \frameworkku{} share the
same idea to test and evaluate only $k$ points, they are based on
completely different criteria for testing.  In \frameworkku{}, all the testings and
formulations are based on \emph{only the higher-priority task
  utilizations}. In \frameworkkq{}, the testings are based \emph{not
  only on the higher-priority task utilizations, but also on the
  higher-priority task execution times}.  The above difference
in the formulations results in completely different properties and
mathematical closed-forms.  The
natural condition of \frameworkkq{} is a \emph{quadratic form} for testing
the schedulability  or the response
time of a task, whereas the natural condition of \frameworkku{} is a
\emph{hyperbolic form} for testing the schedulability of a task.

\emph{If one framework were dominated by another or these two frameworks
were just with minor difference in mathematical formulations, it
wouldn't be necessary to separate and present
them as two different frameworks.} Both frameworks are in fact needed
and have to be applied for different cases.  Here,
we only shortly explain their differences,
advantages, and disadvantages in this paper.  For completeness, another
document has been prepared in
\cite{DBLP:journals/corr/framework-compare} to present the similarity,
the difference and the characteristics of these two frameworks in details.

Since the formulation of \frameworkku{} is more restrictive than
\frameworkkq{}, its applicability is limited by the possibility to
formulate the tests purely by using higher-priority task utilizations
without referring to their execution times. There are cases, in which
formulating the higher-priority interference by using only task
utilizations for \frameworkku{} is troublesome or over-pessimistic. For such cases,
further introducing the upper bound of the execution time by using
\frameworkkq{} is more precise. Most of the presented cases, except the one in
uniprocessor constrained-deadline systems in Appendix
B\citetechreport{} are in the above category.
Although \frameworkkq{} is more general, it is not as precise as
\frameworkku{}, if we can formulate the schedulability tests into both
frameworks with the same parameters.  In such cases, the same
pseudo-polynomial-time (or exponential time) test is used, and the
utilization bound or speed-up factor analysis derived from the \frameworkku{} framework is, in general,
tighter and better.

In a nutshell, \frameworkkq{} is more general, whereas \frameworkku{} is
more precise. If an exact schedulability test can
be constructed and the test can be converted into \frameworkku{},
e.g., uniprocessor scheduling for constrained-deadline task sets,
then, adopting \frameworkku{} leads to tight results. For example, by using
\frameworkkq{}, we can reach the conclusion that the utilization bound
for rate-monotonic scheduling is $2-\sqrt{2} \approx 0.586$, which is
less precise than the Liu and Layland bound $\ln{2}\approx 0.693$,
 a simple implication by using \frameworkku{}. However, if we
are allowed to change the execution time and period of a task for
different job releases (called acyclic task model in
\cite{DBLP:journals/tc/AbdelzaherSL04}), then the tight utilization bound
$2-\sqrt{2}$ can be easily achieved by using \frameworkkq{}. 

Due to the fact the \frameworkku{} is more precise (with respect to the utilization bound) when the exact
tests can be constructed, even though \frameworkku{} is more
restrictive, both are needed for different
cases. Both \frameworkku{} and \frameworkkq{} are general enough to cover a range of spectrum of
applications, ranging from uniprocessor systems to multiprocessor
systems.  For more information and comparisons, please refer to
\cite{DBLP:journals/corr/framework-compare}.

}{
Although the objective is similar to \frameworkku{}, 
\frameworkkq{} in this paper applies completely different criteria from
\frameworkku{} for testing purposes. In \frameworkku{}, all the
testings and formulations are based on \emph{only the higher-priority
  task utilizations}. In \frameworkkq{}, the testings are based
\emph{not only on the higher-priority task utilizations, but also on
  the higher-priority task execution times}.  The above difference in
the formulations results in completely different properties and
mathematical closed-forms.  The natural condition of \frameworkkq{} is
a \emph{quadratic form} for testing the schedulability, whereas the
natural condition of \frameworkku{} is a \emph{hyperbolic form} for
testing the schedulability or the response time of a task.

The  \frameworkkq{} and \frameworkku{} frameworks do not dominate each other, and should be applied for different cases (i.e., some task models are better handled by one framework than the other).\footnote{For completeness, another document has been prepared in \cite{DBLP:journals/corr/framework-compare} to present the similarity, the difference and the characteristics of these two frameworks in details.}
Since the formulation of \frameworkku{} is more restrictive than
\frameworkkq{} (due to using only higher-priority task utilizations
without referring to their execution times), there are cases, in which
formulating the higher-priority interference by using only task
utilizations for \frameworkku{} is troublesome or over-pessimistic. For such cases,
further introducing the upper bound of the execution time by using
\frameworkkq{} is more precise. 
}

\noindent\textbf{Contributions.}
The key contribution of this paper is a general schedulability and
response-time analysis framework, \frameworkkq{}, that can be easily applied to
analyze a number of complex real-time task models, on both
uniprocessor and multiprocessor systems. 
A key novelty of \frameworkkq{} that allows a rather general analysis framework is that we
do not specifically seek for the total utilization bound. Instead, we
look for the critical value in the specified sufficient schedulability
test while verifying the schedulability of task $\tau_k$. This
critical value of task $\tau_k$ gives the difficulty of task $\tau_k$
to be schedulable under the scheduling policy. 
We present several properties of \frameworkkq{}, which
provide a series of closed-form solutions to be adopted
for sufficient tests and worst-case response time analyses for
real-time task models, as long as a corresponding $k$-point last-release
schedulability test (Definition~\ref{def:kpoints}) or a $k$-point last-release
response-time analysis (Definition~\ref{def:kpoints-response}) can be
constructed.  \ifbool{techreport}{The generality of \frameworkkq{} is
supported by demonstrating that either new or better results compared to the
state-of-the-art can be easily obtained using \frameworkkq{}.}{Due to
the space constraint, we are only able to provide some simple examples in this paper. More comprehensive examples and applications can be found in the report in     \cite{DBLP:journals/corr/abs-k2q}. The detailed evaluations,
    compared to other approaches, are in
    \cite{DBLP:journals/corr/framework-compare}.} 
\ifbool{techreport}{
Examples include: 
\begin{itemize}
\item Several utilization-based schedulability and response analyses
  for uniprocessor sporadic task systems are provided in Section
  \ref{sec:sporadic}\ifbool{techreport}{}{ and with more complete
    results in Appendix B in \cite{DBLP:journals/corr/abs-k2q}.}  The
  utilization-based worst-case response-time analysis in
  Theorem~\ref{theorem:response-time-sporadic} in Section
  \ref{sec:sporadic} is identical to the response-time analysis by
  Bini et al. \cite{bini-RTSS2015} developed in parallel.
\item We improve the schedulability tests in multiprocessor global
  fixed-priority scheduling in Appendix C. A
  general condition is a quadratic bound. Specifically, we show that the
  speed-up (capacity augmentation) factor of global RM is
  $\frac{3+\sqrt{7}}{2}\approx 2.823$ for implicit-deadline sporadic
  task systems, which improves upon the existing best speed-up factor $3$ presented in \cite{DBLP:conf/opodis/BertognaCL05}. 
\item We provide, to the best of our knowledge, the first polynomial-time
  worst-case response time analysis for sporadic real-time tasks with
  jitters \cite{DBLP:conf/rtss/BaruahCM97,238595} in
  \ifbool{techreport}{Appendix D.}{Appendix D in \cite{DBLP:journals/corr/abs-k2q}.}
\item We also demonstrate how to convert
  exponential-time schedulability tests of generalized multi-frame
  task models
  \cite{DBLP:journals/rts/BaruahCGM99,DBLP:conf/rtcsa/TakadaS97} to
  polynomial-time tests by using the \frameworkkq{} framework in 
  \ifbool{techreport}{Appendix E.}{Appendix E in \cite{DBLP:journals/corr/abs-k2q}.}
\item The above results are for task-level fixed-priority scheduling
  policies. We further explore mode-level fixed-priority scheduling
  policies by
  studying the acyclic task model
  \cite{DBLP:journals/tc/AbdelzaherSL04} and the multi-mode task model
  \cite{DBLP:conf/rtas/DavisFPS14}.\ifbool{techreport}{\footnote{Although
    the focus in \cite{DBLP:conf/rtas/DavisFPS14} is for
    variable-rate-behaviour tasks, we will refer such a model as a
    multi-mode task model.}}{} We conclude a
  quadratic bound and a utilization bound $2-\sqrt{2}$ for RM
  scheduling policy. The utilization bound is the same as the result in
  \cite{DBLP:journals/tc/AbdelzaherSL04}. They can be further
  generalized to handle more generalized task models,
including the digraph task model \cite{DBLP:conf/rtas/StiggeEGY11},
the recurring real-time task model \cite{DBLP:conf/rtss/Baruah10}. 
   This is presented in
  \ifbool{techreport}{Appendix F.}{Appendix F in \cite{DBLP:journals/corr/abs-k2q}.}
\end{itemize}
}{
}

\ifbool{techreport}{The emphasis of this paper is to show the generality of the
\frameworkkq{} framework by demonstrating via several task models. The
tests and analytical results in the framework are with low
complexity, but can still be shown to provide good results through
speed-up factor or utilization bound analyses.  We also note a somehow surprising finding through developing this framework: 
 \emph{in the past, exponential-time schedulability tests were typically
  not recommended and most of time ignored, as this requires very high
  complexity. We have successfully shown in this paper that
  exponential-time schedulability tests may lead to good
  polynomial-time tests (almost automatically) by using the
  \frameworkkq{} framework. Therefore, this framework may also open the
possibility to re-examine some tests with exponential-time complexity to improve their 
applicability.}}{}


\section{Basic Task and Scheduling Models}
\label{sec:model}

This section presents the sporadic real-time task model, as the basis
for our presentations. Even though the framework targets at more
general task models, to ease the presentation flow, we will start with
the sporadic task models.  A sporadic task $\tau_i$ is released
repeatedly, with each such invocation called a job. The $j^{th}$ job
of $\tau_i$, denoted $\tau_{i,j}$, is released at time $r_{i,j}$ and
has an absolute deadline at time $d_{i,j}$. Each job of any task
$\tau_i$ is assumed to have execution time $C_i$. Here in this paper,
whenever we refer to the execution time of a job, we mean for the
worst-case execution time of the job, since all the analyses we use are safe by only considering the worst-case execution time.  The response time of a job is
defined as its finishing time minus its release time. Successive jobs
of the same task are required to be executed in sequence. Associated with
each task $\tau_i$ are a period $T_i$, which specifies the minimum
time between two consecutive job releases of $\tau_i$, and a deadline
$D_i$, which specifies the relative deadline of each such job, i.e.,
$d_{i,j}=r_{i,j}+D_i$. The worst-case response time of a task $\tau_i$
is the maximum response time among all its jobs. 
The utilization of a task $\tau_i$ is defined
as $U_i=C_i/T_i$.

A sporadic task system $\tau$ is an implicit-deadline
system if $D_i = T_i$ holds for each $\tau_i$. A sporadic task system
$\tau$ is a constrained-deadline system if $D_i \leq T_i$
holds for each $\tau_i$.  Otherwise, such a sporadic task system
$\tau$ is an arbitrary-deadline system.

A task is said \emph{schedulable} by a scheduling policy if all of its
jobs can finish before their absolute deadlines, i.e., the worst-case
response time of the task is no more than its relative deadline.  A
task system is said \emph{schedulable} by a scheduling policy if all
the tasks in the task system are schedulable. A \emph{schedulability
  test} expresses sufficient schedulability conditions to ensure the feasibility
of the resulting schedule by a scheduling policy.

Throughout the paper, we will focus on fixed-priority preemptive
scheduling. That is, each task is associated with a priority level\ifbool{techreport}{ (except in Appendix F).}{.}
For a uniprocessor system, the scheduler always dispatches the job
with the highest priority in the ready queue to be executed.  For a
multiprocessor system, we consider multiprocessor global scheduling on
$M$ identical processors, in which each of them has the same
computation power. For global multiprocessor scheduling, there is a
global queue and a global scheduler to dispatch the jobs. We consider
only global fixed-priority scheduling. At any time, the
$M$-highest-priority jobs in the ready queue are dispatched and
executed on these $M$ processors.

Note that the framework is not only limited to the above task and platform models. These terminologies are introduced only for the simplicity of presentation and illustrating some examples. \ifbool{techreport}{}{The applications with other platform and task models can be found\citetechreport{}.}

\ifbool{techreport}{
\noindent{\bf Speed-Up Factor and Capacity Augmentation Factor:}
To quantify the error of the schedulability tests or
the scheduling policies, the concept of resource augmentation by using
speed-up factors \cite{Phillips:stoc97} and the capacity augmentation
factors \cite{Li:ECRTS14} has been adopted.  For example, global DM in
general does not have good utilization bounds to schedule a set of
sporadic tasks on $M$ identical processors, due to ``Dhall's effect'' \cite{doi:10.1287/opre.26.1.127}.  However, if we constrain
the total utilization $\sum_{\tau_i} \frac{C_i}{M T_i} \leq
\frac{1}{b}$, the density $\frac{C_k+\sum_{\tau_i \in hp(\tau_k)}
  C_i}{M D_k} \leq \frac{1}{b}$ for each task $\tau_k$, and the
maximum utilization $\max_{\tau_i} \frac{C_i}{\min\{T_i, D_i\}} \leq
\frac{1}{b}$, it is possible to provide the schedulability guarantee
of global RM by setting $b$ to $3-\frac{1}{M}$
\cite{DBLP:conf/rtss/AnderssonBJ01,DBLP:conf/rtss/Baker03,DBLP:conf/opodis/BertognaCL05}. Such
a factor $b$ has been recently named as a \emph{capacity augmentation
  factor} \cite{Li:ECRTS14}.  Note that the capacity augmentation
bound was defined without taking this simple condition
$\frac{C_k+\sum_{\tau_i \in hp(\tau_k)} C_i}{M D_k} \leq \frac{1}{b}$
in \cite{Li:ECRTS14}, as they focus on implicit-deadline systems. For
constrained-deadline systems, adding such a new constraint is a
natural extension.

An algorithm ${\cal A}$ is with speed-up factor $b$: \emph{ If there exists a feasible schedule for the task
    system, it is schedulable by algorithm ${\cal A}$ by speeding up
    (each processor) to $b$ times as fast as in the original
    platform (speed).}
A sufficient schedulability test for scheduling
algorithm ${\cal A}$ is with speed-up factor $b$:
\emph{If the task system cannot pass the sufficient
    schedulability test, the task set is not schedulable by any
    scheduling algorithm if (each processor) is slowed down to
    $\frac{1}{b}$ times of the original platform speed.}
Note that if the
capacity augmentation factor is $b$, the speed-up factor is also
upper-bounded by $b$.
}{}

\vspace{-2mm}
\section{Analysis Flow}
\label{sec:flow}

The framework focuses on testing the schedulability and the response
time for a task $\tau_k$, under the assumption that the
required properties (i.e., worst-case response time or the
schedulability) of the higher-priority tasks are already verified and
provided. We will implicitly assume that all the higher-priority tasks
are already verified and the required properties are already obtained.
Therefore, this framework has to be applied for each of the given
tasks. To ensure whether a task system is schedulable by the given
scheduling policy, the test has to be applied for all the tasks.  Of
course, the results can be extended to test the schedulability of a
task system in linear time complexity or to allow on-line admission
control in constant time complexity if the schedulability condition
(or with some more pessimistic simplifications) is monotonic. Such
extensions are presented only for trivial cases.

We will only present the schedulability test of a certain task
$\tau_k$, that is analyzed, under the above assumption. For
notational brevity, in the framework presentation, we will implicitly
assume that there are $k-1$ tasks, say $\tau_1, \tau_2, \ldots,
\tau_{k-1}$ with higher-priority than task $\tau_k$. We will use
$hp(\tau_k)$ to denote the set of these $k-1$ higher-priority tasks,
when their orderings do not matter. Moreover, we only
consider the cases when $k \geq 2$, since $k=1$ is pretty trivial.

\section{\frameworkkq{}}
\label{sec:framework}

This section presents the basic properties of the \frameworkkq{}
framework for testing the schedulability of task $\tau_k$ in a given
set of real-time tasks (depending on the specific models given in each
application). Before presenting the
framework, we first give a simple example to explain the underlying
concepts by using an
implicit-deadline sporadic task system $\tau$, in which $D_i =
T_i$ for every $\tau_i \in \tau$. The exact schedulability test to
verify whether task $\tau_k$ can meet its deadline under
fixed-priority scheduling on uniprocessor systems is to check
\begin{equation}
\label{eq:TDA-implicit}
\exists t \mbox{ with } 0 < t \leq T_k {\;\; and \;\;} C_k +
\sum_{\tau_i \in hp(\tau_k)} \ceiling{\frac{t}{T_i}}C_i \leq t,
\end{equation}
where $hp(\tau_k)$ is the set of tasks with higher priority than
$\tau_k$. Instead of testing all the time points $t$ in the range of
$0$ and $T_k$, for a sufficient schedulability test, we can greedily
only consider to test the time points $(\ceiling{\frac{T_k}{T_i}}-1)T_i$ for
$\tau_i \in hp(\tau_k)$ and $t=T_k$. If $ C_k + \sum_{\tau_i \in
  hp(\tau_k)} \ceiling{\frac{t}{T_i}}C_i \leq t$ holds in one of those
$k$ tested time points, then we can conclude that $\tau_k$ can be feasibly
scheduled under this scheduling policy.

To implement to above testing concept, we need two definitions: 1)
Definition \ref{def:lease-release} defines the last release time
ordering so that we can formulate the problem with linear algebra, 2)
Definition \ref{def:kpoints} defines an abstracted schedulability test
that can be used to model general schedulability tests regardless of
the task and platform model.
\begin{definition}[Last Release Time Ordering]
\label{def:lease-release}
Let $\pi$ be the last release time ordering assignment as a bijective
function $\pi:  hp(\tau_k)
\rightarrow \setof{1, 2, \ldots,k-1}$ to define the last release time ordering
of task $\tau_j \in hp(\tau_k)$ in the window of interest. Last release time orderings are
numbered from $1$ to $k-1$, i.e., $|hp(\tau_k)|$, where 1 is the earliest and $k-1$ the
latest. \myendproof
\end{definition}

The last release time ordering is a very important property in the
whole framework. When testing the schedulability or analyzing the
worst-case response time of task $\tau_k$, we do not need the
priority ordering of the higher-priority tasks in
$hp(\tau_k)$. But, we need to know how to order the $k-1$ higher-priority tasks so that we can formulate the test with simple and linear arithmetics based on the total order.
For the rest of this paper, the ordering of the $k-1$ higher-priority tasks implicitly refers to their last release time ordering (except explanations regarding the last release time ordering when referring to Example \ref{example-3tasks}).
In the \frameworkkq{} framework, we are only interested
to test only $k$ time points. More precisely, we are only interested
to test whether task $\tau_k$ can be successfully executed before the
last release time of a higher-priority task in the testing
window. Therefore, the last release time ordering provides a total
order so that we can transform the schedulability tests into the following definition.

\begin{definition}
  \label{def:kpoints}
  
  A $k$-point last-release schedulability test under a given last release time ordering
  $\pi$ of the $k-1$ higher-priority tasks is a sufficient
  schedulability test of a fixed-priority scheduling policy, that verifies the existence of  $t_j$ with $j=1,2,\ldots,k$ such that 
  $0 \leq t_1 \leq t_2 \leq \cdots \leq t_{k-1}  \leq t_k$ and \begin{equation}
    \label{eq:precodition-schedulability}
    C_k + \sum_{i=1}^{k-1} \alpha_i t_i U_i + \sum_{i=1}^{j-1} \beta_i C_i \leq t_j, 
  \end{equation}
  where $C_k > 0$, for $i=1,2,\ldots,k-1$, $\alpha_i > 0$, $U_i > 0$, $C_i \geq 0$, and $\beta_i >0$ are dependent upon the setting
  of the task models and task $\tau_i$. \myendproof
\end{definition}

\begin{example}
  \label{example-1} {\bf Implicit-deadline task systems}: For an
  implicit-deadline sporadic task system $\tau$, suppose that we are
  interested to test whether task $\tau_k$ can meet its deadline or
  not under a fixed-priority scheduling algorithm on a uniprocessor
  platform. Let $|hp(\tau_k)|$ be $k-1$ and the tasks in $hp(\tau_k)$ be
  ordered by $(\ceiling{\frac{T_k}{T_i}}-1)T_i$ non-decreasingly,
  i.e., $t_1 \leq t_2 \leq \cdots \leq t_{k-1} \leq t_k = T_k$.  For a
  specific testing point at time $t_j$ for a certain $j=1,2,\ldots,k$,
  the function $\ceiling{\frac{t_j}{T_i}}C_i$ (to quantify the
  workload due to the jobs released by a higher-priority task $\tau_i \in
  hp(\tau_k)$) has two cases: 1) if $i < j$, due to the definition of
  $t_i$ as $(\ceiling{\frac{T_k}{T_i}}-1)T_i$ and $t_i \leq t_j \leq
  T_k$, we know that $\ceiling{\frac{t_j}{T_i}}C_i$ is upper bounded
  by $\ceiling{\frac{t_i}{T_i}}C_i+C_i = t_i U_i + C_i$; 2) if $i \geq
  j$, due to the definition of $t_i$ as
  $(\ceiling{\frac{T_k}{T_i}}-1)T_i$ and $t_j \leq t_i \leq T_k$, we
  know that $\ceiling{\frac{t_j}{T_i}}C_i$ is upper bounded by
  $\ceiling{\frac{t_i}{T_i}}C_i = t_i U_i$.\footnote{Since $t_i$ is an integer multiple of $T_i$, the property $\ceiling{\frac{t_i}{T_i}}C_i = t_i U_i$ holds.}

  By the above analysis, for a given $j=1,2,\ldots,k$, we know that
  $C_k + \sum_{i=1}^{k-1} \ceiling{\frac{t_j}{T_i}}C_i \leq C_k +
  \sum_{i=1}^{k-1} t_i U_i + \sum_{i=1}^{j-1} C_i$.  Therefore, we
  know that task $\tau_k$ is schedulable by the fixed-priority
  scheduling if there exists $j \in \setof{1,2,\ldots,k}$ such that
  \begin{align*}
    \small
    C_k + \sum_{i=1}^{k-1} t_i U_i + \sum_{i=1}^{j-1}  C_i \leq t_j.
  \end{align*}\normalsize  
  In other words, by the specific index rule of the tasks in
  $hp(\tau_k)$ and setting $\alpha_i=1$ and $\beta_i=1$ for every task
  $\tau_i$ in $hp(\tau_k)$, we reach a concrete example for
  Definition~\ref{def:kpoints}. \endproof
\end{example}


A concrete example is provided  here
for illustrating Example~\ref{example-1}.
\begin{example}\label{example-3tasks-v0}
  Consider that $k=3$ and $|hp(\tau_k)|$ is $2$. For the two tasks in
  $hp(\tau_k)$, let $C_1=2, U_1=0.2, T_1=10$ and $C_2=4, U_2=0.5,
  T_2=8$.  Suppose that $t_3=D_3=T_3=36$.  By the transformation in
  Example~\ref{example-1}, we know that $t_1=30$ and $t_2=32$. The last release time ordering $\pi$ of $\setof{\tau_1, \tau_2}$ follows the index, i.e., $\pi:\setof{\tau_1, \tau_2} \rightarrow \setof{1, 2}$. Moreover, $\alpha_1=\alpha_2=\beta_1=\beta_2=1$.
\hfill\myendproof
\end{example}

Similar to Definition~\ref{def:kpoints}, we can also define an
abstracted worst-case response time analysis as follows:
\begin{definition}
  \label{def:kpoints-response}
  A $k$-point last-release response time analysis is a safe response time
  analysis of a fixed-priority scheduling policy under a given last release time
  ordering $\pi$ of the $k-1$ higher-priority tasks by finding the
  maximum
  \begin{equation}
    \label{eq:precond-objective-k}
   t_k = C_k + \sum_{i=1}^{k-1} \alpha_i t_i U_i + \sum_{i=1}^{k-1} \beta_i C_i,
  \end{equation}
with $0 \leq t_1 \leq t_2 \leq \cdots \leq t_{k-1} \leq t_{k}$ and
  \begin{align}
    \label{eq:precond-objective-k2}
    C_k + \sum_{i=1}^{k-1} \alpha_i t_i U_i + \sum_{i=1}^{j-1} \beta_i C_i > t_j, & \forall j=1,2,\ldots,k-1,
  \end{align}
  where $C_k > 0$, $\alpha_i > 0$, $U_i > 0$, $C_i \geq 0$, and $\beta_i >0$ are dependent upon the setting
  of the task models and task $\tau_i$. \myendproof
\end{definition}

\begin{example}
  \label{example-response-time} {\bf Response-time for
    constrained-deadline task systems}: Suppose that $R_k$ is the
  exact worst-case response time for task $\tau_k$ and $R_k \leq T_k$ under uniprocessor
  fixed-priority scheduling. That is, by Eq.~\eqref{eq:TDA-implicit},
  $C_k + \sum_{\tau_i \in hp(\tau_k)} \ceiling{\frac{t}{T_i}}C_i > t$
  for any $0 < t < R_k$ and $C_k + \sum_{\tau_i \in hp(\tau_k)}
  \ceiling{\frac{R_k}{T_i}}C_i = R_k$. Similar to
  Example~\ref{example-1}, let $|hp(\tau_k)|$ be $k-1$ and the tasks in
  $hp(\tau_k)$ be ordered by $(\ceiling{\frac{R_k}{T_i}}-1)T_i$
  non-decreasingly, i.e., $t_1 \leq t_2 \leq \cdots \leq t_{k-1} \leq
  R_k$. With the same analysis in Example~\ref{example-1}, we know
  that $ C_k + \sum_{i=1}^{k-1} t_i U_i + \sum_{i=1}^{j-1} C_i > t_j$
  for $j=1,2,\ldots,k-1$ and $R_k \leq C_k + \sum_{i=1}^{k-1} t_i U_i
  + \sum_{i=1}^{k-1} C_i$. As a result, by the specific index rule of
  the tasks in $hp(\tau_k)$ and setting $\alpha_i=1$ and $\beta_i=1$
  for every task $\tau_i$ in $hp(\tau_k)$, we reach a concrete example
  for Definition~\ref{def:kpoints-response}.
\endproof
\end{example}


\subsection{Important Notes} 
Before presenting the analyses based on Definition~\ref{def:kpoints}
and Definition~\ref{def:kpoints-response}, we would like to first
explain the important assumptions and the flow to use the analytical
results.  Throughout the paper, we implicitly assume that $t_k > 0$
when Definition~\ref{def:kpoints} is used. Moreover, we
only consider non-trivial cases, in which $C_k >0$ and $0 < U_i \leq
1$ for $i=1,2,\ldots,k-1$. The definition of $t_k$
depends on how Definition~\ref{def:kpoints} is constructed based on
the original schedulability test, usually equal to the length of the
interval (of the points to be tested in the original schedulability test),
e.g., $t_k=T_k=D_k$ in Example~\ref{example-1}.  In most of the cases,
we can set $t_k$ as $D_k$. \ifbool{techreport}{
But, it can also be set to other cases, to
be demonstrated in Appendix C for global RM scheduling.  
}{But, it can also be set to other cases, which can be found\citetechreport{}. }

In Definition~\ref{def:kpoints}, the $k$-point last-release
schedulability test is a sufficient schedulability test that tests
only $k$ time points, defined by the $k-1$ higher-priority tasks and
task $\tau_k$. 
Similarly,  in Definition~\ref{def:kpoints-response}, a $k$-point
last-release response time analysis provides a safe response time
by only testing whether task $\tau_k$ has already finished earlier at
$k-1$ points, each defined by a higher-priority task.

In both cases in Definitions~\ref{def:kpoints} and
\ref{def:kpoints-response}, the last release time ordering $\pi$ is
assumed to be given. In some cases, this ordering can be easily
obtained. For such cases, all the lemmas in this section can be directly adopted.
However, in most of the cases in our demonstrated task
models, we have to test all possible last release time orderings and
take the worst case. Fortunately, we will show that finding the
worst-case ordering is not a difficult problem, which requires to sort
the $k-1$ higher-priority tasks under a simple criteria, in Lemmas~\ref{lemma:general-sorting} and~\ref{lemma:general-response-sorting}. 
Therefore, for such cases, the lemmas in this section have to be adopted by combining with Lemma~\ref{lemma:general-sorting}~or~\ref{lemma:general-response-sorting}.

We first assume that the corresponding coefficients $\alpha_i$ and
$\beta_i$ in
Definitions~\ref{def:kpoints}~and~\ref{def:kpoints-response} are
given. How to derive them will be discussed in the following sections.
Clearly their values are highly dependent upon the task models and the
scheduling policies.  Provided that these coefficients $\alpha_i$,
$\beta_i$, $C_i$, $U_i$ for every higher-priority task $\tau_i \in
hp(\tau_k)$ are given, we analyze (1) the response time by finding the
extreme case for a given $C_k$ (under
Definition~\ref{def:kpoints-response}), or (2) the schedulability by
finding the extreme case for a given $C_k$ and $D_k$. Therefore, the
\frameworkkq{} framework provides utilization-based schedulability
analyses and response time analyses automatically if the
corresponding parameters $\alpha_i$ and $\beta_i$ can be defined to
ensure that the tests in Definitions~\ref{def:kpoints}
and~\ref{def:kpoints-response} are safe.

\frameworkkq{} can be used by a wide range of applications, as long
as the users can properly specify the corresponding task properties
$C_i$ and $U_i$ and the constant coefficients $\alpha_i$ and $\beta_i$
of every higher-priority task $\tau_i$. More precisely, the formulation in
Definitions~\ref{def:kpoints} and \ref{def:kpoints-response} does not
actually care what $C_i$ and $U_i$ actually mean. When sporadic
task models are considered, we will use these two terms as they were
defined in Section~\ref{sec:model}, i.e., $C_i$ stands for the
execution time and  $U_i$ is $\frac{C_i}{T_i}$. When we consider more
general cases, such as the generalized multi-frame and multi-mode task
models, we have to properly define the values of $U_i$ and $C_i$ to
apply the framework.

The use cases of \frameworkkq{} can be achieved by using
the known schedulability tests (that are in the form of pseudo
polynomial-time or exponential-time tests) or some simple modifications of the existing
results. \ifbool{techreport}{We will provide the explanations of the correctness of the
selection of the parameters, $\alpha_i, \beta_i, C_i, U_i$ for a
higher-priority task $\tau_i$ to support the correctness of the
results.}{}  
Such a flow actually leads to the elegance and the
generality of the framework, which works as long as
Definition~\ref{def:kpoints} (Definition~\ref{def:kpoints-response},
respectively) can be successfully constructed for the sufficient
schedulability test (response time, respectively) of task $\tau_k$ in
a fixed-priority scheduling policy. The procedure is illustrated in Figure~\ref{fig:framework}.
With the availability of the \frameworkkq{} framework, the quadratic
bounds or utilization bounds can be automatically derived as
long as the safe upper bounds $\alpha$ and $\beta$ can be safely
derived, regardless of the task model or the platforms. 

\ifbool{techreport}{
We are not going to present how to
\emph{systematically and automatically} derive these parameters to be applied for the
\frameworkkq{} framework. For most of the typical schedulability tests
and response time analyses in real-time systems, such a derivation procedure is similar to the
automatic parameter generation for the \frameworkku{} in \cite{DBLP:journals/corr/abs-k2u-automatic}.
}
{
}

\begin{figure}[t]
	\begin{center}
      \begin{tikzpicture}[scale=0.72,transform shape]
        \path \practica {1}{\underline{\bf Demonstrated Applications:}
          \begin{tabular}{ll}
          Sec. 5:& Arbitrary-deadline sporadic tasks\\
          Sec. 5:& Multiprocessor RM\\
          App. D\ifbool{techreport}{}{ \cite{DBLP:journals/corr/abs-k2q}}:&
          Periodic tasks with jitters\\
          App. E\ifbool{techreport}{}{ \cite{DBLP:journals/corr/abs-k2q}}:&
          Generalized multiframe\\
          App. F\ifbool{techreport}{}{ \cite{DBLP:journals/corr/abs-k2q}}:&
          Acyclic and Multi-Mode Models\\
          \end{tabular}
        };
        \path (p1.west)+(2.2,-3.8) node(p2)[materia, rounded
        rectangle,text width=7em]{{\bf $U_i, \forall i < k$\\ $C_i,
            \forall i < k$\\$\alpha_i, \forall i<k$\\$\beta_i, \forall
            i <k$\\$C_k$\\ $t_k$ {\scriptsize(for Lemmas \ref{lemma:framework-general-schedulability}-\ref{lemma:framework-totalU-constrained})}}}; 
        \path (p2.north)+(1.5,0.45) node[text width=16em]{Derive
          parameters\\by \underline{Definitions 2 or 3}};
        \path (p2.east)+(2,0) node(p3)[practica,fill=blue!20,text width=5em,text centered]{\large{\bf $k^2Q$\\ framework}};

        \path (p3.east)+(0.4,+4) node(p4)[practica,fill=green!30,minimum width=6em,text width=5em,text centered]{Quadratic bound}; 
        \path (p4.south)+(2.3,-1.5) node(p5)[practica,fill=green!30,minimum width=6em,text width=5em,text centered]{Other utilization bounds}; 
        \path (p5.south)+(0,-1.65)
        node(p6)[practica,fill=green!30,minimum width=6em,text
        width=5em,text centered]{Response-time test}; 
        
        \path [line, ->] (p1.west) -- +(-0.3,0) node[black, rotate=90,
        yshift=0.3cm, xshift=-1.6cm]{\footnotesize{define the
            $\pi$ ordering or use Lemma 2 or 7}} -- + (-0.3, -3.6) -- (p2.west);
        \path [line, ->] (p2.east) -- (p3.west);
        \path [line, ->] (p3.east)+(0,0.3) -- node[rotate=80,yshift=0.3cm,black]{Lemma \ref{lemma:framework-general-schedulability}} (p4.south);
        \path [line, ->] (p3.east)+(0,0) -- node[rotate=60,yshift=0.3cm,black]{Lemmas \ref{lemma:framework-constrained-schedulability}-\ref{lemma:framework-totalU-constrained}} (p5.west);
        \path [line, ->] (p3.east)+(0,-0.3) -- node[rotate=5,yshift=0.3cm,black]{Lemma \ref{lemma:framework-general-response}} (p6.west);
      \end{tikzpicture}    
	\end{center}
\vspace{-2mm}
\caption{The \frameworkkq{} framework. }
\label{fig:framework}
\end{figure}
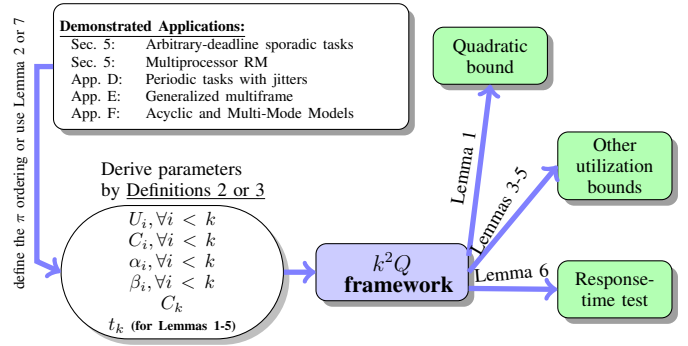

\subsection{Schedulability Test Framework}

This section provides five important lemmas for deriving the
utilization-based schedulability test based on
Definition~\ref{def:kpoints}. Lemma~\ref{lemma:framework-general-schedulability}
is the most general test, whereas
Lemmas~\ref{lemma:framework-constrained-schedulability},
\ref{lemma:framework-totalU-exclusive},~and~\ref{lemma:framework-totalU-constrained} work for certain special
cases when $\beta_i C_i \leq \beta U_i t_k$ for any higher-priority
task $\tau_i$. Lemma~\ref{lemma:general-sorting} gives the worst-case last release time ordering, 
which can be used when the last release time ordering for testing task $\tau_k$ is unknown.
  
\begin{lemma}
\label{lemma:framework-general-schedulability}
For a given $k$-point last-release schedulability test, defined in
Definition~\ref{def:kpoints}, of a scheduling 
algorithm,
in which $0 < \alpha_i$, and $0 < \beta_i$ for any
$i=1,2,\ldots,k-1$, $0 < t_k$, $\sum_{i=1}^{k-1}\alpha_i U_i \leq 1$, and $\sum_{i=1}^{k-1}
\beta_i C_i \leq t_k$, task $\tau_k$ is schedulable by the
fixed-priority scheduling
algorithm if the following condition holds
\begin{equation}
\label{eq:schedulability-general}
\frac{C_k}{t_k} \leq 1 - \sum_{i=1}^{k-1}\alpha_i U_i - \frac{\sum_{i=1}^{k-1} (\beta_i C_i - \alpha_i U_i (\sum_{\ell=i}^{k-1}  \beta_\ell C_\ell) )}{t_k}.
\end{equation}
\end{lemma}

\begin{proof}
  We prove this lemma by showing that the condition in
  Eq.~\eqref{eq:schedulability-general} leads to the satisfactions of the
  schedulability conditions listed in
  Eq.~(\ref{eq:precodition-schedulability}) by using contrapositive.
  By taking the negation of the schedulability condition in
  Eq.~(\ref{eq:precodition-schedulability}), we know that if task
  $\tau_k$ is \emph{not schedulable} by the scheduling policy, then
  for each $j=1,2,\ldots, k$
  \begin{equation}
    \label{eq:lp-init-constraints}
    C_k + \sum_{i=1}^{k-1} \alpha_i t_i U_i + \sum_{i=1}^{j-1} \beta_i C_i > t_j.  
  \end{equation}

  To enforce the condition in
  Eq.~\eqref{eq:lp-init-constraints}, we are going to
  show that $C_k$ must have some lower bound, denoted as $C_k^*$. Therefore, if $C_k$ is
  no more than this lower bound, then task $\tau_k$ is schedulable by the
  scheduling policy.  For the rest of the proof, we replace $>$ with $\geq$ in
  Eq.~\eqref{eq:lp-init-constraints}, as the infimum and the minimum
  are the same when presenting the inequality with $\geq$.
 The unschedulability for satisfying Eq.~\eqref{eq:lp-init-constraints} implies that $C_k > C_k^*$, where $C_k^*$ is defined in the optimization problem:

\begin{subequations}\label{eq:lp-init0}
\footnotesize  \begin{align}
    \mbox{min\;\;} & C_k^*\\
    \mbox{s.t.\;\;} &     C_k^* + \sum_{i=1}^{k-1} \alpha_i t_i^*
    U_i + \sum_{i=1}^{j-1} \beta_i C_i \geq t_j^*,&\forall
    j=1,2,\ldots,
    k-1,   \\
  &t_1^* \geq 0 & \\
  &t_j^* \geq t_{j-1}^*,&\forall j=2,3,\ldots, k-1,\\
    & C_k^* + \sum_{i=1}^{k-1} \alpha_i t_i^* U_i + \sum_{i=1}^{k-1} \beta_i C_i \geq t_k,&   
  \end{align}    
  \end{subequations}
  where $t^*_1, t^*_2, \ldots, t^*_{k-1}$ and $C_k^*$ are variables,
  $\alpha_i$, $\beta_i$, $U_i$, and $C_i$ are constants, and $t_k$ is
  a given positive constant.  Moreover, it is obvious that relaxing
  the constraint $t_{j}^* \geq t_{j-1}^*$ for $j=2,3,\ldots,k-1$ by
  using $t_j^* \geq 0$ does not increase the corresponding objective
  function in the linear programming. 
Therefore, we have

  \begin{subequations}\label{eq:lp-init}
\footnotesize  \begin{align}
    \mbox{min\;\;} & C_k^* \label{eq:precodition-objective}\\
    \mbox{s.t.\;\;} &     C_k^* + \sum_{i=1}^{k-1} \alpha_i t_i^*
    U_i + \sum_{i=1}^{j-1} \beta_i C_i \geq t_j^*,&\forall
    j=1,2,\ldots,
    k-1,     \label{eq:precodition-schedulability-negation-0}\\
  &t_j^* \geq 0,&\forall j=1,2,\ldots, k-1,     \label{eq:precodition-schedulability-negation-1}\\
    & C_k^* + \sum_{i=1}^{k-1} \alpha_i t_i^* U_i + \sum_{i=1}^{k-1} \beta_i C_i \geq t_k.&     \label{eq:precodition-schedulability-negation-2}
  \end{align}    
  \end{subequations}

  Let $s \geq 0$ be a slack variable such that $C_k^* = t_k + s - (\sum_{i=1}^{k-1} \alpha_i t_i^*U_i + \sum_{i=1}^{k-1} \beta_i C_i)$.
  Therefore, we can replace the objective function and the constraints
  with the above equality of $C_k^*$. The objective function (i.e.,
  Eq.~\eqref{eq:precodition-objective}) is
  to find the minimum value of $t_k +s - (\sum_{i=1}^{k-1} \alpha_i
  t_i^*U_i + \sum_{i=1}^{k-1} \beta_i C_i)$ such that
  Eq.~\eqref{eq:precodition-schedulability-negation-0} holds, which is
  equivalent to

{\footnotesize
  \begin{align}
   & t_k + s - (\sum_{i=1}^{k-1} \alpha_i t_i^*U_i + \sum_{i=1}^{k-1} \beta_i C_i) + \sum_{i=1}^{k-1} \alpha_i t_i^*U_i + \sum_{i=1}^{j-1} \beta_i C_i\nonumber\\
    =\; & t_k+ s- \sum_{i=j}^{k-1} \beta_i C_i 
    \geq t_j^*, \;\;\;\forall j=1,2,\ldots, k-1. \label{eq:precodition-schedulability-negation}
  \end{align}
  }



 For notational brevity, let $t_k^*$ be $t_k+s$.
 Therefore, the linear programming in Eq.~\eqref{eq:lp-init} can be rewritten as follows:
\begin{subequations}
    \label{eq:lp-framework-general}
  \begin{align}\footnotesize
    \mbox{min } &  
    t_k^* - (\sum_{i=1}^{k-1} \alpha_i U_it_i^* + \sum_{i=1}^{k-1}\beta_i C_i)\\
    \mbox{s.t.} \;\;&
t_k^*- \sum_{i=j}^{k-1} \beta_i C_i \geq t_j^*,
      & \forall 1 \leq j \leq k - 1, \label{eq:lp-framework-general-constraints}\\
&t_j^* \geq 0
      & \forall 1 \leq j \leq k - 1. \label{eq:lp-framework-general-boundaryconstraints}\\
&t_k^* \geq t_k \label{eq:lp-framework-general-slack-constraints}
  \end{align}    
  \end{subequations}
  
  The remaining proof is to solve the above linear programming to obtain the minimum $C_k^*$.
Our proof strategy is to solve the linear programming analytically as a function of $t_k^*$. This can be imagined as if $t_k^*$ is given. At the end, we will prove the optimality by considering all possible $t_k^* \geq t_k$.
  This involves three steps:
  \begin{itemize}
  \item Step 1: we analyze certain properties of optimal solutions based on the extreme point theorem for linear programming \cite{luenberger2008linear} under the assumption that $t_k^*$ is given as a constant, i.e., $s$ is known.
  \item Step 2: we present a specific solution in an \emph{extreme point}, as a function of $t_k^*$.
  \item Step 3: we prove that the above extreme point solution gives the minimum $C_k^*$  if $\sum_{i=1}^{k-1}\alpha_i U_i \leq 1$.
  \end{itemize}

 {\bf [Step 1:]} After specifying the value
  $t_k^*$ as a given constant, the new linear programming without the
  constraint in
  Eq.~\eqref{eq:lp-framework-general-slack-constraints} has only
  $k-1$ variables and $2(k-1)$ constraints.
Thus, according to the extreme
  point theorem for linear programming \cite{luenberger2008linear},
  the linear constraints form a polyhedron of feasible solutions. The
  extreme point theorem states that either there is no feasible
  solution or one of the extreme points in the polyhedron is an
  optimal solution when the objective of the linear programming is
  finite. To satisfy Eqs.~\eqref{eq:lp-framework-general-constraints} and~\eqref{eq:lp-framework-general-boundaryconstraints}, we know that 
  $t_j^* \leq t_k^*$ for $j=1,2,\ldots,k-1$, due to $t_i^* \geq 0$, $0 < \beta $, and $C_i \geq 0$ for $i=j, j+1, \ldots,k-1$. As a result, the objective of the above linear programming is finite 
  since a feasible solution has to satisfy $t_i^* \leq t_k^*$ for $i=1,2,\ldots,k-1$,.

According to the extreme point theorem, one of the extreme points is
the optimal solution of Eq.~\eqref{eq:lp-framework-general}.  There
are $k-1$ variables with $2k-2$ constraints in
Eq.~\eqref{eq:lp-framework-general}. An extreme point must have at
least $k-1$ \emph{active} constraints in
Eqs.~\eqref{eq:lp-framework-general-constraints}
and~\eqref{eq:lp-framework-general-boundaryconstraints}, in which
their $\geq$ are set to equality $=$.

{\bf [Step 2:]} 
One special extreme point solution by setting $t_j^* > 0$ is to put $t_k^* -  \sum_{i=j}^{k-1} \beta_i C_i = t_j^*$ for every $j=1,2,\ldots,k-1$, i.e.,
\begin{equation}
\label{eq:periodrelation}
\forall 1 \leq i \leq k - 1, \quad t_{i+1}^* - t_i^* = \beta_i C_i,
\end{equation}
which implies that
\begin{equation}
\label{eq:3rdperiodrelation}
t_k^* - t_i^* = \sum_{\ell=i}^{k-1} (t_{\ell+1}^*- t_\ell^*) = \sum_{\ell=i}^{k-1}  \beta_\ell C_\ell
\end{equation}
The above extreme point solution is always feasible in the linear programming due to the assumption that  
$\sum_{j=1}^{k-1}  \beta_j C_j \leq t_k \leq t_k^*$.
  Therefore,  in this extreme point solution, the objective function of Eq.~\eqref{eq:lp-framework-general} by rephrasing based on the condition in Eq.~\eqref{eq:3rdperiodrelation} is 
  {\small \begin{align}
    \label{eq:fina-first-lemma}
& t_k^* - \sum_{i=1}^{k-1} \left(\alpha_i U_i t_i^* + \beta_i C_i \right) \\
= & t_k^* - \sum_{i=1}^{k-1} \left(\alpha_i U_i \left(t_k^*-\sum_{\ell=i}^{k-1}  \beta_\ell C_\ell\right) + \beta_i C_i \right)\\
= & t_k^* - \left(\sum_{i=1}^{k-1}\alpha_i U_i t_k^* + \sum_{i=1}^{k-1} \beta_i C_i -\sum_{i=1}^{k-1} \alpha_i U_i \left(\sum_{\ell=i}^{k-1}  \beta_\ell C_\ell\right) \right)
  \end{align}}
which means that $C_k^* \geq  t_k^* (1-\sum_{i=1}^{k-1}\alpha_i U_i) - \sum_{i=1}^{k-1} (\beta_i C_i - \alpha_i U_i (\sum_{\ell=i}^{k-1}  \beta_\ell C_\ell) )$.

{\bf [Step 3:]}
The rest of the proof shows that other feasible extreme point
solutions (that allow $t_j^*$ to be $0$ for some higher-priority task
$\tau_j$) are with worse objective values for
Eq.~\eqref{eq:lp-framework-general}.  Under the assumption that
$\sum_{i=1}^{k-1} \beta_i C_i \leq t_k\leq t_k^*$, if $t_j^*$ is set to $0$,
there are two cases: (1) $t_k^*- \sum_{i=j}^{k-1} \beta_i C_i > 0$ or
(2) $t_k^*- \sum_{i=j}^{k-1} \beta_i C_i = 0$. In the former case, we
can simply set $t_j^*$ to $t_k^*- \sum_{i=j}^{k-1} \beta_i C_i$ to
improve the objective function without introducing any violation of
the constraints. In the latter case, the value of $t_j^*$ can only be
set to $0$ in any feasible solutions. Therefore, we conclude that any
other feasible extreme point solutions for Eq.~\eqref{eq:lp-framework-general}
are worse.

Note that the above solution of $C_k^*$ is still a function of $t_k^*$. We  need to find the minimization of $C_k^*$ with respect to $t_k^*$ based on the fact $t_k^* \geq t_k$.
Due to the assumption that $1-\sum_{i=1}^{k-1}\alpha_i U_i \geq 0$ and $t_k^* \geq t_k$, we know that $t_k^* (1-\sum_{i=1}^{k-1}\alpha_i U_i) \geq t_k (1-\sum_{i=1}^{k-1}\alpha_i U_i)$. Therefore, $C_k^* =  t_k (1-\sum_{i=1}^{k-1}\alpha_i U_i) - \sum_{i=1}^{k-1} (\beta_i C_i - \alpha_i U_i (\sum_{\ell=i}^{k-1}  \beta_\ell C_\ell) )$ when $1-\sum_{i=1}^{k-1}\alpha_i U_i \geq 0$ and $\sum_{i=1}^{k-1} \beta_i C_i \leq t_k$, which concludes the proof.
\end{proof}

Lemma~\ref{lemma:framework-general-schedulability} can be applied only
when the last release time ordering of the $k-1$ higher-priority tasks
is given. We demonstrate the importance of the last release time
ordering by using the following example.\footnote{To demonstrate the impact of the last release time ordering, we use the original task indexes before applying $\pi_1$ or $\pi_2$ whenever referring to Example~\ref{example-3tasks}.}

\begin{example}\label{example-3tasks}
  Consider that $k=3$ and $|hp(\tau_k)|$ is $2$. For the two tasks in
  $hp(\tau_k)$, let $C_1=2, U_1=0.2, T_1=10$ and $C_2=4, U_2=0.5,
  T_2=8$.  Suppose that $t_3=D_3=T_3=36$.  By the transformation in
  Example~\ref{example-1}, we know that $\alpha_i=1$ and $\beta_i=1$
  for $i=1,2$. 

There are two last release time orderings. Suppose that  
$\pi_1:\setof{\tau_1, \tau_2} \rightarrow \setof{1, 2}$ and
$\pi_2:\setof{\tau_1, \tau_2} \rightarrow \setof{2, 1}$. That is, the
last release time ordering is $\tau_1, \tau_2$ in $\pi_1$, and the
last release time ordering is $\tau_2, \tau_1$ in $\pi_2$.
\hfill\myendproof
\end{example}

Now, we can use Lemma~\ref{lemma:framework-general-schedulability}
based on $\pi_1$ and $\pi_2$:
\begin{compactitem}
\item For $\pi_1$, the schedulability condition in
  Lemma~\ref{lemma:framework-general-schedulability} shows that task
  $\tau_3$ in Example~\ref{example-3tasks} can meet the deadline if $C_3 \leq t_3\cdot(1-U_1-U_2) -
  (C_1-U_1(C_1+C_2) + C_2-U_2C_2)=0.3t_3 - 2.8 = 8$.
\item For $\pi_2$, the schedulability condition in
  Lemma~\ref{lemma:framework-general-schedulability} shows that task
  $\tau_3$ in Example~\ref{example-3tasks} can meet the deadline if $C_3 \leq t_3\cdot(1-U_2-U_1) -
  (C_2-U_2(C_2+C_1) + C_1-U_1C_1)=0.3t_3 - 2.6 = 8.2$.
\end{compactitem}

The immediate question is whether both  $C_3 \leq 8$ based on
$\pi_1$ and $C_3 \leq 8.2$ based on $\pi_2$ are safe. When $t_k=36$,
the transformation in Example~\ref{example-1} in fact adopts the last
release time ordering $\pi_1$. Therefore,
Lemma~\ref{lemma:framework-general-schedulability} is only safe under
$\pi_1$ in this example. As a result, the test in
Lemma~\ref{lemma:framework-general-schedulability} for the above
example is only valid when we apply $\pi_1$.

However, in practice, we usually do not know how these tasks are indexed according to the required last release in the window of interest. It may seem at first glance that we need to test all the possible orderings. Fortunately, with the following lemma, we can safely consider only one specific last release time ordering of the $k-1$ higher-priority tasks. 
\begin{lemma}
  \label{lemma:general-sorting}
  The worst-case ordering $\pi$ of the $k-1$ higher-priority tasks under the schedulability condition in Eq.~\eqref{eq:schedulability-general} in Lemma~\ref{lemma:framework-general-schedulability} is to order the tasks in a non-increasing order of $\frac{\beta_i C_i}{\alpha_i U_i}$,
 in which $0 < \alpha_i$ and $0 < \beta_i$ for any $i=1,2,\ldots,k-1$, $0 < t_k$.
\end{lemma}
\begin{proof}
This lemma is proved by showing that the schedulability condition in
  Lemma~\ref{lemma:framework-general-schedulability}, i.e., $1 -
  \sum_{i=1}^{k-1}\alpha_i U_i - \frac{\sum_{i=1}^{k-1} \beta_i
    C_i}{t_k} + \frac{\alpha_i U_i (\sum_{\ell=i}^{k-1} \beta_\ell C_\ell)}{t_k}$, is minimized, when the $k-1$ higher-priority tasks are indexed in a non-increasing order of $\frac{\beta_i C_i}{\alpha_i U_i}$.  Suppose that there are two adjacent tasks $\tau_{h}$ and $\tau_{h+1}$ with $\frac{\beta_h C_h}{\alpha_h U_h} < \frac{\beta_{h+1} C_{h+1}}{\alpha_{h+1} U_{h+1}}$.  Let us now examine the difference of $\frac{\sum_{i=1}^{k-1} \alpha_i U_i (\sum_{\ell=i}^{k-1} \beta_\ell C_\ell) }{t_k}$ by swapping the index of task $\tau_{h}$ and task $\tau_{h+1}$.

  It can be easily observed that the other tasks $\tau_i$ with $i\neq
  h$ and $i\neq h+1$ do not change their corresponding values
  $\alpha_i U_i (\sum_{\ell=i}^{k-1} \beta_\ell C_\ell)$ in both
  orderings (before and after swapping $\tau_{h}$ and $\tau_{h+1}$).
  The difference in the term $\alpha_h U_h (\sum_{\ell=h}^{k-1}
  \beta_\ell C_\ell) + \alpha_{h+1} U_{h+1} (\sum_{\ell=h}^{k-1}
  \beta_\ell C_\ell) $ before and after swapping tasks $\tau_{\ell}$
  and $\tau_{\ell+1}$  (before - after)  is 
  \begin{align*}\footnotesize
& \left( (\alpha_h U_h \beta_{h+1} C_{h+1} - \alpha_{h+1} U_{h+1} \beta_{h} C_{h}\right) \\
 = & \alpha_h\alpha_{h+1}U_h U_{h+1} \left( \frac{\beta_{h+1} C_{h+1}}{\alpha_{h+1} U_{h+1}}-\frac{\beta_h C_h}{\alpha_h U_h}  \right) > 0.
  \end{align*}\normalsize
Therefore, we reach the conclusion that swapping $\tau_h$ and $\tau_{h+1}$ in the ordering makes the schedulabilty condition more stringent. By applying the above swapping repetitively, we reach the conclusion that ordering the tasks 
in a non-increasing order of $\frac{\beta_i C_i}{\alpha_i U_i}$ has the most stringent schedulability condition in Eq.~\eqref{eq:schedulability-general}.  
\end{proof}

We again use the configuration in Example~\ref{example-3tasks} to
demonstrate the rationale behind Lemma~\ref{lemma:general-sorting}. In
this example, let us consider that $t_3=T_3=23$. When $t_k=23$, the
transformation in Example~\ref{example-1} in fact adopts the last
release time ordering $\pi_2$, i.e., $\tau_3$ is schedulable if $C_3
\leq 0.3t_3-2.6=4.3$. The schedulability condition based on the last
release time ordering $\pi_1$, i.e., $\tau_3$ is schedulable if $C_3
\leq 0.3t_3-2.8=4.1$, is always worse than that based on $\pi_2$ by
Lemma~\ref{lemma:general-sorting}. Therefore, it is always safe to use
$\pi_1$, even though it can be sometimes more pessimistic, e.g., when
$t_3$ is $23$.

\subsection{Different Utilization Bounds}
The analysis in Lemma~\ref{lemma:framework-general-schedulability} uses the execution time and the utilization of the tasks in $hp(\tau_k)$ to build an upper bound of $C_k/t_k$ for  schedulability tests. It is also very convenient in real-time systems to build schedulability tests only based on utilization of the tasks. We explain how to achieve that in the following lemmas under the assumptions that $0 < \alpha_i \leq \alpha$, and $0 < \beta_i C_i \leq \beta U_i t_k$ for any $i=1,2,\ldots,k-1$. 
These lemmas are useful when we are interested to derive utilization bounds, speed-up factors, resource augmentation factors, etc., for a given scheduling policy by defining the coefficients $\alpha$ and $\beta$ according to the scheduling policies independently from the detailed parameters of the tasks. 
 Since the property repeats in all the statements, we make a formal definition before presenting the lemmas.
\begin{definition}
  \label{def:alpha-upper-bound}
  Lemmas~\ref{lemma:framework-constrained-schedulability} to
  \ref{lemma:framework-totalU-constrained} are based on the
  following $k$-point last-release schedulability test of a scheduling
  algorithm, defined in Definition~\ref{def:kpoints}, in which $0 <
  \alpha_i \leq \alpha$, and $0 < \beta_i C_i \leq \beta U_i t_k$ for
  any $i=1,2,\ldots,k-1$, $0 < t_k$, $\alpha\sum_{i=1}^{k-1}U_i \leq
  1$, and $\beta\sum_{i=1}^{k-1}U_i \leq 1$.
\end{definition}

\begin{lemma}
\label{lemma:framework-constrained-schedulability}
For a given $k$-point last-release schedulability test of a scheduling 
algorithm, with the properties in Definition~\ref{def:alpha-upper-bound}, 
task $\tau_k$ is schedulable by the scheduling
algorithm if the following condition holds
{\small \begin{align}
\label{eq:schedulability-constrained}
\frac{C_k}{t_k} \leq &1 - (\alpha+\beta)\sum_{i=1}^{k-1} U_i + \alpha\beta\sum_{i=1}^{k-1} U_i (\sum_{\ell=i}^{k-1}  U_\ell)\\
= &1 - (\alpha+\beta)\sum_{i=1}^{k-1} U_i + 0.5\alpha\beta\left((\sum_{i=1}^{k-1} U_i)^2 +(\sum_{i=1}^{k-1} U_i^2)\right)
\label{eq:schedulability-constrained-2}
\end{align}}
\end{lemma}
\begin{proof}
  The condition in Eq.~\eqref{eq:schedulability-constrained} comes by reformulating the proof of Lemma~\ref{lemma:framework-general-schedulability} with $\beta U_i t_k^*$ instead of $\beta_i C_i$. All the procedures remain the same, and, therefore, $\beta_i C_i$ for task $\tau_i$ in the right-hand side of Eq.~\eqref{eq:schedulability-general} can be replaced by $\beta U_i$.

   We focus on the condition in Eq.~\eqref{eq:schedulability-constrained-2} by showing that $\sum_{i=1}^{k-1} U_i (\sum_{\ell=i}^{k-1}  U_\ell) = 0.5\left((\sum_{i=1}^{k-1} U_i)^2 +(\sum_{i=1}^{k-1} U_i^2)\right)$. This condition clearly holds when $k=2$ since $U_1^2 = 0.5(U_1^2+U_1^2)$. We consider $k \geq 3$. This is due to
   \begin{align*}&   \sum_{i=1}^{k-1} U_i (\sum_{\ell=i}^{k-1}  U_\ell) =      \sum_{i=1}^{k-1} U_i^2 + \sum_{i=1}^{k-2} U_i (\sum_{\ell=i+1}^{k-1}  U_\ell) \\
=^1\;\; &\sum_{i=1}^{k-1} U_i^2 + 0.5\left(\left(\sum_{i=1}^{k-1} U_i\right)^2 -  \sum_{i=1}^{k-1} U_i^2\right)\\
=\;\;\; & 0.5\left((\sum_{i=1}^{k-1} U_i)^2 +(\sum_{i=1}^{k-1} U_i^2)\right),
   \end{align*}
where $=^1$ follows from the fact $\sum_{i=1}^{k-2} U_i (\sum_{\ell=i+1}^{k-1}  U_\ell) = \sum_{i=2}^{k-1} U_i (\sum_{\ell=1}^{i-1}  U_\ell) = 0.5\left(\left(\sum_{i=1}^{k-1} U_i\right)^2 -  \sum_{i=1}^{k-1} U_i^2\right)$.
\end{proof}

Lemma~\ref{lemma:framework-constrained-schedulability} provides a
schedulability test based on a quadratic form by using only the
utilization of the higher-priority tasks with the properties in
Definition~\ref{def:alpha-upper-bound}. The following two lemmas are
applicable for testing the utilization bound(s), i.e., the summation
of the task utilization. 

\begin{lemma}
\label{lemma:framework-totalU-exclusive}
For a given $k$-point last-release schedulability test of a scheduling 
algorithm, with the properties in Definition~\ref{def:alpha-upper-bound}, 
task $\tau_k$ is schedulable by the scheduling
algorithm if 
\begin{equation}\small
\label{eq:schedulability-totalU-exclusive}
\sum_{i=1}^{k-1}U_i \leq \left(\frac{k-1}{k}\right)\left( \dfrac{\alpha+\beta-\sqrt{(\alpha+\beta)^2-2 \alpha\beta (1-\frac{C_k}{t_k})\frac{k}{k-1}}}{\alpha\beta}\right).
\end{equation}
\end{lemma}
\begin{proof}
  This can be formally proved by using the Lagrange Multiplier
  Method. However, it can also be proved by using a simpler
  mathematical observation. Suppose that $x=\sum_{i=1}^{k-1} U_i$ is
  given. For given $\alpha, \beta,$ and $x$, we know that
  Eq.~\eqref{eq:schedulability-constrained-2} becomes
  $1-(\alpha+\beta)x+0.5\alpha\beta (x^2+\sum_{i=1}^{k-1}U_i^2)$. That
  is, only the last term $0.5\alpha\beta (\sum_{i=1}^{k-1}U_i^2)$
  depends on how $U_i$ values are actually assigned. Moreover,
  $\sum_{i=1}^{k-1} U_i^2$ is a well-known convex function with
  respect to $U_1, U_2, \ldots, U_{k-1}$. That is, $\rho U_i^2 +
  (1-\rho) U_j^2 \geq \left(\rho U_i+ (1-\rho)U_j\right)^2$ for any $0
  \leq \rho \leq 1$. Therefore, $\sum_{i=1}^{k-1} U_i^2$ is minimized
  when $U_1=U_2=\cdots=U_{k-1}=\frac{x}{k-1}$.

Hence, what we have to do is to find the infimum $x$ such that the
condition in Eq.~\eqref{eq:schedulability-constrained-2} does not
hold. That is,
\begin{align*}
  \mbox{infimum } & x \\
  \mbox{s. t. } &\frac{C_k}{t_k} > 1- (\alpha+\beta)x+0.5\alpha\beta \left(x^2 +
\frac{x^2}{k-1}\right).
\end{align*}
 This means that as long as $\sum_{i=1}^{k-1} U_i$ is no more
than such infimum $x$, the condition in
Eq.~\eqref{eq:schedulability-constrained-2} always holds and the
schedulability can be guaranteed. 
Provided that $\frac{C_k}{t_k}$ is given, we can simply solve the above problem by finding the $x$ with 
$0 = 1-\frac{C_k}{t_k} - (\alpha+\beta)x+0.5\alpha\beta \frac{k}{k-1} x^2$. There are two roots in the above quadratic equation. 
The smaller root, i.e., the right-hand side of Eq.~\eqref{eq:schedulability-totalU-exclusive}, is the infimum by definition.
\end{proof}

\begin{lemma}
\label{lemma:framework-totalU-constrained}
For a given $k$-point last-release schedulability test of a scheduling 
algorithm, with the properties in Definition~\ref{def:alpha-upper-bound}, 
provided that $\alpha+\beta \geq 1$, 
then task $\tau_k$ is schedulable by the scheduling
algorithm if  
{\small
\begin{align}
&\frac{C_k}{t_k} + \sum_{i=1}^{k-1} U_i \leq \nonumber\\
&\begin{cases}\label{eq:schedulability-totalU-constrained}
  \left(\frac{k-1}{k}\right)\left( \dfrac{\alpha+\beta-\sqrt{(\alpha+\beta)^2-2 \alpha\beta \frac{k}{k-1}}}{\alpha\beta}\right), &
  \begin{array}{l}
    \mbox{ if  } k > \frac{(\alpha+\beta)^2-1}{\alpha^2+\beta^2-1} \\    
    \mbox{ and  }\alpha^2+\beta^2 > 1
  \end{array}\\
1 + \frac{(k-1)((\alpha+\beta-1) - \frac{1}{2}(\alpha+\beta)^2+0.5) }{k\alpha\beta} & \mbox{ otherwise}
\end{cases}
 \end{align} 
}
\end{lemma}
\begin{proof}
  The proof is similar to the proof of Lemma~\ref{lemma:framework-totalU-exclusive}, but slightly more involved. We detail the proof in Appendix A.
\end{proof}

\ifbool{techreport}{
By the fact that $\sqrt{(\alpha+\beta)^2-2 \alpha\beta \frac{k}{k-1}}
= \sqrt{(\alpha+\beta)^2-2 \alpha\beta - 2\alpha\beta\frac{1}{k-1}}$,
which is an increasing function with respect to $k$, and the fact that
$\frac{k-1}{k}$ is a decreasing function with respect to $k$, we know
that the right-hand side of
Eq.~\eqref{eq:schedulability-totalU-constrained} (when
$\alpha^2+\beta^2 > 1$) decreases with
respect to $k$. Similarly, the right-hand side of
Eq.~\eqref{eq:schedulability-totalU-exclusive} also decreases with
respect to $k$. Therefore, for evaluating the utilization bounds, it
is alway safe to take $k\rightarrow \infty$ as a safe upper bound. The
right-hand side of Eq.~\eqref{eq:schedulability-totalU-exclusive}
converges to $\frac{\alpha+\beta-\sqrt{\alpha^2+\beta^2+2
    \alpha\beta\frac{C_k}{t_k}}}{\alpha\beta}$ when $k\rightarrow
\infty$.  The right-hand side of
Eq.~\eqref{eq:schedulability-totalU-constrained}  (when
$\alpha^2+\beta^2 > 1$) converges to
$\frac{\alpha+\beta-\sqrt{\alpha^2+\beta^2}}{\alpha\beta}$ when
$k\rightarrow \infty$.
}

\subsection{Response Time Analysis Framework}

We now further discuss the utilization-based response-time analysis
framework. 
\begin{lemma}
\label{lemma:framework-general-response}
For a given $k$-point response time analysis, defined in
Definition~\ref{def:kpoints-response}, of a scheduling 
algorithm,
in which $0 < \alpha_i \leq \alpha$, $0 < \beta_i \leq \beta$ for any
$i=1,2,\ldots,k-1$, $0 < t_k$ and $\sum_{i=1}^{k-1}
\alpha_i U_i < 1$, the response time to execute $C_k$ for task
$\tau_k$ is at most
\begin{equation}
\label{eq:schedulability-general-response}
\frac{C_k+ \sum_{i=1}^{k-1} \beta_i C_i - \sum_{i=1}^{k-1} \alpha_i U_i (\sum_{\ell=i}^{k-1} \beta_{\ell} C_{\ell})}{1-\sum_{i=1}^{k-1} \alpha_i U_i}.
\end{equation}
\end{lemma}
\begin{proof}
  The proof is similar to the proof of
  Lemma~\ref{lemma:framework-general-schedulability}.  The detailed
  proof is in Appendix A.
\end{proof}

We use the same example in Example~\ref{example-3tasks} by setting
$C_3=8$ to demonstrate how to use
Lemma~\ref{lemma:framework-general-response}.  By the transformation
in Example~\ref{example-response-time}, we know that $\alpha_i=1$ and
$\beta_i=1$ for $i=1,2$.  Now, we can use
Lemma~\ref{lemma:framework-general-response} based on $\pi_1$ and
$\pi_2$ (defined in Example~\ref{example-3tasks}) to calculate the
worst-case response time:
\begin{compactitem}
\item For $\pi_1$, the response-time analysis in
  Lemma~\ref{lemma:framework-general-response} shows that the response
  time of task $\tau_3$ in Example~\ref{example-3tasks} is upper
  bounded by $\frac{C_3 + C_1+C_2 - U_1(C_1+C_2) - U_2C_2}{1-U_1-U_2}
  = \frac{14-0.2\times 6 - 0.5 \times 4}{0.3} = 36$.
\item For $\pi_2$, the response-time analysis in
  Lemma~\ref{lemma:framework-general-response} shows that the response
  time of task $\tau_3$ in Example~\ref{example-3tasks} is upper
  bounded by $\frac{C_3 + C_1+C_2 - U_2(C_2+C_1) - U_1C_1}{1-U_1-U_2}
  = \frac{14-0.5\times 6 - 0.2 \times 2}{0.3} = 35\frac{1}{3}$. 
\end{compactitem}

Not all the last release time orderings are safe for the worst-case
response time analysis. Fortunately, similar to
Lemma~\ref{lemma:general-sorting}, we can safely consider only one
specific last release time ordering of the $k-1$ higher-priority tasks
as shown in the following lemma.
\begin{lemma}
  \label{lemma:general-response-sorting}
  The worst-case ordering $\pi$ of the $k-1$ higher-priority tasks
  under the response bound in
  Eq.~\eqref{eq:schedulability-general-response} in Lemma~\ref{lemma:framework-general-response} is to order the tasks in a non-increasing order of $\frac{\beta_i C_i}{\alpha_i U_i}$,
 in which $0 < \alpha_i$ and $0 < \beta_i$ for any $i=1,2,\ldots,k-1$, $0 < t_k$.
\end{lemma}
\begin{proof}
  The ordering of the $k-1$ higher-priority tasks in the indexing rule
  only matters for the term $\sum_{i=1}^{k-1} \alpha_i U_i
  (\sum_{\ell=i}^{k-1} \beta_{\ell} C_{\ell})$, which was already
  proved in the proof of Lemma~\ref{lemma:general-sorting} to be
  minimized by ordering the tasks in a non-increasing order of
  $\frac{\beta_i C_i}{\alpha_i U_i}$. Clearly, the minimization of
  $\sum_{i=1}^{k-1} \alpha_i U_i
  (\sum_{\ell=i}^{k-1} \beta_{\ell} C_{\ell})$ also leads to the maximization of
  Eq.~\eqref{eq:schedulability-general-response}, which concludes
  the proof.
\end{proof}

As a result, thanks to the help of
Lemma~\ref{lemma:general-response-sorting}, we can conclude that
$\pi_1$ in the example in this subsection is a safe last release time
ordering to use Lemma~\ref{lemma:framework-general-response} for the
worst-case response time analysis.

\ifbool{techreport}{
\begin{table*}[t]
\renewcommand{\arraystretch}{1.2}
  \centering
  \scalebox{0.8}{
\begin{tabular}{|p{8cm}|c|c|p{4cm}|}
    \hline
    \hline
    Model & $\alpha_i$ & $\beta_i$ & c.f.\\ 
    \hline
    \hline
  Uniprocessor Sporadic Tasks& $\alpha_i=1$ & $\beta_i = 1$ &
  Theorems \ref{theorem:schedulability-sporadic-arbitrary} and
  \ref{theorem:response-time-sporadic}\\
    \hline
   Multiprocessor Global RM/DM for Sporadic Tasks & $\alpha_i=\frac{1}{M}$ &
   $\beta_i = \frac{1}{M}$ &
   Theorems~\ref{thm:multiprocessor-GRM-M-1-carry-k2q},~\ref{thm:multiprocessor-GRM-M-1-carry-k2q-V2},~\ref{thm:multiprocessor-grm-sporadic-tight},and~\ref{thm:multiprocessor-gdm-sporadic-tight}\\
    \hline 
  Uniprocessor Periodic Tasks with Jitters& $\alpha_i=1$ & $\beta_i = 1$ &
  Theorem~\ref{theorem:response-time-PJ}\\
    \hline
  Uniprocessor Generalized Multi-Frame, Acyclic, and Mode-Change
  Tasks& $\alpha_i=1$ & $\beta_i = 1$ & 
  Theorems~\ref{theorem:response-time-sporadic-gmf},
  \ref{theorem:schedulability-mode-change}, and \ref{theorem:schedulability-mode-change-rm}.\\
    \hline 
    \hline
  \end{tabular}
}
  \caption{\small The $\alpha_i$ and $\beta_i$ parameters in our demonstrated task models.}
  \label{tab:alpha-beta}
\end{table*}
}{}

\section{Applications by Using Sporadic Task Models}
\label{sec:sporadic}

This section demonstrates how to use the \frameworkkq{} framework to
derive utilization-based schedulability and response-time
analyses for sporadic task systems in uniprocessor and multiprocessor
systems. As sporadic real-time task models are the simplest scenarios
that can demonstrate how to use \frameworkkq{}, the content here is
merely for explaining how to use the framework, but not for
demonstrating the generality or superiority of \frameworkkq{}.

\subsection{Uniprocessor Constrained-Deadline Systems} 

\begin{theorem}
\label{theorem:sporadic-constrained-pessimistic}
Task $\tau_k$ in a sporadic task system with constrained deadlines is
schedulable by the fixed-priority scheduling algorithm if
$\sum_{i=1}^{k-1}\frac{C_i}{D_k} \leq 1$ and
\begin{equation}
\frac{C_k}{D_k} \leq 1-\sum_{i=1}^{k-1}U_i-\sum_{i=1}^{k-1}\frac{C_i}{D_k}+\frac{\sum_{i=1}^{k-1} U_i (\sum_{\ell=i}^{k-1}  C_\ell)}{D_k},
\end{equation}
in which the $k-1$ higher-priority tasks in $hp_1(\tau_k)$ are indexed
in a non-increasing order of $T_i$.
\end{theorem}
\begin{proof}
  This comes from
  Lemma~\ref{lemma:framework-general-schedulability}~and~\ref{lemma:general-sorting}
  based on the setting $\alpha_i=1$ and $\beta_i=1$ to satisfy
  Definition~\ref{def:kpoints}.\footnote{If $T_i > D_k$ for a certain
    task $\tau_i$ in $hp(\tau_k)$, we can simply set $t_i$ to $0$.}
\end{proof}

\begin{theorem}
\label{theorem-rm}
Task $\tau_k$ in a sporadic constrained-deadline task system with is
schedulable by the rate-monotonic (RM) scheduling algorithm if
\begin{equation}\label{eq:rm-qb0}
 \frac{C_k}{D_k} \leq 1- 2\sum_{i=1}^{k-1}U_i +0.5\left((\sum_{i=1}^{k-1} U_i)^2 +(\sum_{i=1}^{k-1} U_i^2)\right)
\end{equation}
or
\begin{equation}\label{eq:rm-qb1}
 \sum_{i=1}^{k-1}U_i \leq \left(\frac{k-1}{k}\right)\left(2-\sqrt{4-\frac{2k (1- \frac{C_k}{D_k} )}{k-1}}\right)
\end{equation} 
or
\begin{equation}\label{eq:rm-qb2}
\frac{C_k}{D_k} + \sum_{i=1}^{k-1}U_i \leq 
  \begin{cases}
\left(\frac{k-1}{k}\right)\left(2-\sqrt{4-\frac{2k}{k-1}}\right) &\mbox{ if } k > 3\\ 
1 -\frac{k-1 }{2k} &\mbox{ if } k \leq 3
  \end{cases}
\end{equation}
\end{theorem}
\begin{proof}
  Under RM scheduling, we know that $C_i = U_i T_i \leq U_i
  T_k$. Therefore, $\alpha$ can be set to $1$ and $\beta$
  can be set to $1$ in
  Definition~\ref{def:alpha-upper-bound}. Eq.~\eqref{eq:rm-qb0} is due
  to Lemma~\ref{lemma:framework-constrained-schedulability},
  Eq.~\eqref{eq:rm-qb1} is due to
  Lemma~\ref{lemma:framework-totalU-exclusive}, and
  Eq.~\eqref{eq:rm-qb2} is due to
  Lemma~\ref{lemma:framework-totalU-constrained}.
\end{proof}

\ifbool{techreport}{
The above result in Theorem~\ref{theorem-rm} leads to the utilization
bound $2-\sqrt{2}$ for implicit-deadline sporadic task systems under
RM scheduling.  This analysis is less precise than the Liu and Layland
bound $\ln{2}\approx 0.693$, a simple implication by using
\frameworkku{}. However, if we are allowed to change the execution
time and period of a task for different job releases (called acyclic
task model in \cite{DBLP:journals/tc/AbdelzaherSL04}), then the tight
utilization bound $2-\sqrt{2}$ can be easily achieved by using
\frameworkkq{}, detailed in Appendix F\citetechreport{}.
}
{
}
\subsection{Uniprocessor Arbitrary-Deadline Systems}  

\ifbool{techreport}{ 
For a specified fixed-priority scheduling algorithm, let $hp(\tau_k)$
be the set of tasks with higher priority than $\tau_k$. We now
classify the task set $hp(\tau_k)$ into two subsets:
\begin{itemize}
\item $hp_1(\tau_k)$ consists of the higher-priority tasks with periods
  smaller than $D_k$.
\item $hp_2(\tau_k)$ consists of the higher-priority tasks with periods
  larger than or equal to $D_k$.
\end{itemize}
}
{
}

The
exact schedulability analysis for arbitrary-deadline task sets under
fixed-priority scheduling has been developed in
\cite{DBLP:conf/rtss/Lehoczky90}. The schedulability analysis is to
use a \emph{busy-window} concept to evaluate the worst-case response
time. That is, we release all the higher-priority tasks together with
task $\tau_k$ at time $0$ and all the subsequent jobs are released as
early as possible by respecting to the minimum inter-arrival time. The
busy window finishes when a job of task $\tau_k$ finishes before the
next release of a job of task $\tau_k$. It has been shown in
\cite{DBLP:conf/rtss/Lehoczky90} that the worst-case response time of
task $\tau_k$ can be found in one of the jobs of task $\tau_k$ in the
busy window. 

For the $h$-th job of task $\tau_k$ in the busy window, the finishing
time $R_{k,h}$ is the minimum $t$ such that
\[ 
h C_k + \sum_{i=1}^{k-1} \ceiling{\frac{t}{T_i}}C_i \leq t, 
\] 
and, hence, its response time is $R_{k,h}-(h-1)T_k$. The busy window
of task $\tau_k$ finishes on the $h$-th job if $R_{k,h} \leq h
T_k$.

\ifbool{techreport}{ 
We can create a virtual
sporadic task $\tau_k'$ with execution time $C_k' =
\ceiling{\frac{D_k}{T_k}}C_k + \sum_{\tau_i \in hp_2(\tau_k)} C_i$,
relative deadline $D_k'=D_k$, and period $T_k'=D_k$. 
For notational brevity, suppose that there are $k^*-1$ tasks in
$hp_1(\tau_k)$. 
We have then the
following theorem. 

\begin{theorem}
\label{theorem:schedulability-sporadic-arbitrary}
Task $\tau_k$ in a sporadic task system is
schedulable by the fixed-priority scheduling algorithm if
$\sum_{i=1}^{k^*-1}\frac{C_i}{D_k} \leq 1$ and
\begin{equation}
\label{eq:schedulability-sporadic-any-arbitrary-a}
\frac{C_k'}{D_k} \leq 1-\sum_{i=1}^{k^*-1}U_i-\sum_{i=1}^{k^*-1}\frac{C_i}{D_k}+\frac{\sum_{i=1}^{k^*-1} U_i (\sum_{\ell=i}^{k^*-1}  C_\ell)}{D_k},
\end{equation}
in which $C_k' = \ceiling{\frac{D_k}{T_k}}C_k + \sum_{\tau_i \in
  hp_2(\tau_k)} C_i$, and the $k^*-1$ higher-priority tasks in $hp_1(\tau_k)$ are indexed
in a non-decreasing order of $\left(\ceiling{\frac{D_k}{T_i}}-1\right)T_i$.
\end{theorem}
\begin{proof}
  The analysis is based on the observation to test whether the busy
  window can finish within interval length $D_k$, which was also
  adopted in \cite{conf:/rtns09/Davis} and
  \cite{ChenHLRTSS2015}. By setting $t_i =\left
    (\ceiling{\frac{D_k}{T_i}}-1\right)T_i$, and indexing the tasks in
  a non-decreasing order of $t_i$ leads to the satisfaction of
  Definition~\ref{def:kpoints} with $\alpha_i=1$ and $\beta_i=1$.
\end{proof}

Analyzing the schedulability by using
Theorem~\ref{theorem:schedulability-sporadic-arbitrary} can be good if
$\frac{D_k}{T_k}$ is small. However, as the busy window may be
stretched when $\frac{D_k}{T_k}$ is large, we further present how to
safely estimate the worst-case response time. 
}{
}
Suppose that $t_j = \left(\ceiling{\frac{R_{k,h}}{T_j}}-1\right)T_j$
for a higher-priority task $\tau_j$. We index the tasks such that the
last release ordering $\pi$ of the $k-1$ higher-priority tasks is with
$t_j \leq t_{j+1}$ for $j=1,2,\ldots,k-2$. Therefore, we know that
$R_{k,h}$ is upper bounded by finding the maximum
  \begin{equation}
    \label{eq:sporadic-arbitrary-objective-k}
   t_k = hC_k + \sum_{i=1}^{k-1} t_i U_i + \sum_{i=1}^{k-1} C_i,
  \end{equation}
with $0 \leq t_1 \leq t_2 \leq \cdots \leq t_{k-1} \leq t_{k}$ and
  \begin{align}
    \label{eq:sporadic-arbitrary-objective-k2}
    hC_k + \sum_{i=1}^{k-1}  t_i U_i + \sum_{i=1}^{j-1}  C_i > t_j, & \forall j=1,2,\ldots,k-1.
  \end{align}
Therefore, the above derivation of $R_{k,h}$ satisfies
Definition~\ref{def:kpoints-response} with $\alpha_i=1$, and
$\beta_i=1$ for any higher-priority task $\tau_i$. However, it should
be noted that the last release time ordering $\pi$ is actually unknown
since $R_{k,h}$ is unknown. Therefore, we have to apply
Lemma~\ref{lemma:general-response-sorting} for such cases to obtain
the worst-case release time ordering, i.e., the $k-1$ higher-priority tasks are ordered
  in a non-increasing order of their periods.

\begin{lemma}
  \label{lemma:finishing-time-sporadic-h}
  Suppose that $\sum_{i=1}^{k-1}  U_i \leq 1$. Then, for any
  $h \geq 1$ and $C_k > 0$, we have
  \begin{equation}
    \label{eq:R-k-h}
     R_{k,h} \leq \frac{hC_k+ \sum_{i=1}^{k-1}  C_i - \sum_{i=1}^{k-1} U_i (\sum_{\ell=i}^{k-1}  C_{\ell})}{1-\sum_{i=1}^{k-1}  U_i},
  \end{equation}
  where the $k-1$ higher-priority tasks are ordered
  in a non-increasing order of their periods.
\end{lemma}
\begin{proof}
  This comes from the above discussions with $\alpha_i=1$, $\beta_i=1$
  by applying Lemmas~\ref{lemma:framework-general-response} and
  \ref{lemma:general-response-sorting} when $\sum_{i=1}^{k-1} 
  U_i < 1$. The case when $\sum_{i=1}^{k-1} U_i = 1$ has a
  safe upper bound $R_{k,h}=\infty$ in Eq.~\eqref{eq:R-k-h}.
\end{proof}

\begin{theorem}
  \label{theorem:response-time-sporadic}
  Suppose that $\sum_{i=1}^{k} U_i \leq 1$. The worst-case
  response time of task $\tau_k$ is at most
  \begin{equation}
    \label{eq:R-k}
    R_{k} \leq \frac{C_k+ \sum_{i=1}^{k-1}  C_i - \sum_{i=1}^{k-1} U_i (\sum_{\ell=i}^{k-1}  C_{\ell})}{1-\sum_{i=1}^{k-1}  U_i},
  \end{equation}
  where the $k-1$ higher-priority tasks are ordered
  in a non-increasing order of their periods.
\end{theorem}
\begin{proof}
  This can be proved by showing that $R_{k,h}-(h-1)T_k$ is maximized
  when $h$ is $1$, where $R_{k,h}$ is derived by using
  Lemma~\ref{lemma:finishing-time-sporadic-h}.  The first-order
  derivative of $R_{k,h}-(h-1)T_k$ with respect to $h$ is
  $\frac{C_k}{1-\sum_{i=1}^{k-1} U_i} - T_k = \frac{C_k -
    (1-\sum_{i=1}^{k-1} U_i) T_k }{1-\sum_{i=1}^{k-1} U_i}$.
  There are two cases:

\noindent{\bf Case 1:} If $\sum_{i=1}^{k} U_i < 1$, then $\frac{C_k -
    (1-\sum_{i=1}^{k-1} U_i) T_k }{1-\sum_{i=1}^{k-1} U_i} < \frac{C_k -
    U_k T_k }{1-\sum_{i=1}^{k-1} U_i} = 0$. Therefore,
  $R_{k,h}-(h-1)T_k$ is a decreasing function of $h$. Therefore, the
  response time is maximized when $h$ is $1$.

  \noindent{\bf Case 2:} If $\sum_{i=1}^{k} U_i = 1$, then we know
  that $\frac{C_k - (1-\sum_{i=1}^{k-1} U_i) T_k }{1-\sum_{i=1}^{k-1}
    U_i} = 0$. Therefore, $R_{k,h}-(h-1)T_k$ remains the same
  regardless of $h$.

  Therefore, for both cases, the worst-case response time of task
  $\tau_k$ can be safely bounded by Eq.~\eqref{eq:R-k}. Moreover,
  since the worst case happens when $h=1$, we do not have to check the
  length of the busy window, and we reach
  our conclusion.
\end{proof}

\ifbool{techreport}{ 

\begin{corollary}
  \label{corollary:arbitrary-response-schedulability}
  Task $\tau_k$ in a sporadic task system is schedulable by the
  fixed-priority scheduling algorithm if $\sum_{i=1}^{k} U_i \leq 1$
  and 
  \begin{equation}
    \label{eq:R-D-k}
    \frac{C_k}{D_k} \leq 1-\sum_{i=1}^{k-1}  U_i -
    \frac{\sum_{i=1}^{k-1}  C_i}{D_k} + \frac{\sum_{i=1}^{k-1} U_i (\sum_{\ell=i}^{k-1}  C_{\ell})}{D_k},
  \end{equation}
  where the $k-1$ higher-priority tasks are ordered
  in a non-increasing order of their periods.
\end{corollary}

\noindent{\bf Remarks:} 
The utilization-based worst-case response-time analysis in  
Theorem~\ref{theorem:response-time-sporadic} is analytically tighter  
than the best known result, $R_{k} \leq \frac{C_k+ \sum_{i=1}^{k-1}
  C_i - \sum_{i=1}^{k-1} U_i C_i}{1-\sum_{i=1}^{k-1} U_i}$, by Bini et  
al. \cite{bini2009response}. 
 Lehoczky \cite{DBLP:conf/rtss/Lehoczky90} also provides the total
utilization bound of RM scheduling for arbitrary-deadline systems. The
analysis in \cite{DBLP:conf/rtss/Lehoczky90} is based on the Liu and
Layland analysis \cite{liu1973scheduling}. The resulting utilization
bound is a function of $\Delta=\max_{\tau_i}\{\frac{D_i}{T_i}\}$. When
$\Delta$ is $1$, it is an implicit-deadline system. The utilization
bound in \cite{DBLP:conf/rtss/Lehoczky90} has a closed-form when
$\Delta$ is an integer. However, calculating the utilization bound for
non-integer $\Delta$ is done asymptotically for $k=\infty$ with
a complicated analysis. 
}{
\noindent{\bf Some Remarks:} In parallel, Bini et
al. \cite{bini-RTSS2015} have recently also developed a similar
worst-case response time bound, i.e., their Theorem 1 in
\cite{bini-RTSS2015} is very similar to our
Theorem~\ref{theorem:response-time-sporadic}. Although the proofs are
completely different, we reach the same bound. They also show that the
response-time bound is the tightest continuous function upper bounding
the exact response time of sets of tasks with full utilization when
there are only two tasks in the system. Note that we can obtain different
utilization-based tests by exploiting different properties in the
\frameworkkq{} framework. 
}

\subsection{Multiprocessor Implicit-Deadline Systems}  

We now present how to use \frameworkkq{} to analyze the schedulability
for implicit-deadline sporadic task systems under global
rate-monotonic (global RM) scheduling. Here, we start from the
pseudo-polynomial-time schedulability test by Guan et
al. \cite{DBLP:conf/rtss/GuanSYY09} that we only have to consider
$M-1$ tasks with carry-in jobs, for constrained-deadline (hence, also
for implicit-deadline) task sets.  More precisely, we can define two
different time-demand functions, depending on whether task $\tau_i$ is
with a carry-in job or not:\footnote{This is an over-approximation of
  the linear function used by Guan et
  al. \cite{DBLP:conf/rtss/GuanSYY09}.}
\begin{equation}
  \label{eq:W_i-carryin}
W_i^{carry}(t) =
\begin{cases}
  C_i & 0 < t < C_i\\
  C_i + \ceiling{\frac{t-C_i}{T_i}}C_i & otherwise,
\end{cases}
\end{equation}
and
\begin{equation}
  \label{eq:W_i-normal}
W_i^{normal}(t) = \ceiling{\frac{t}{T_i}}C_i.
\end{equation}
Moreover, we can further over-approximate $W_i^{carry}(t)$, since  $W_i^{carry}(t) \leq W_i^{normal}(t)+C_i$. Therefore, a sufficient schedulability test for testing task $\tau_k$ with $k > M$ for global RM is to verify whether 
\begin{equation}
  \label{eq:grm-multiprocessor-M-1-carryin}
\exists 0 < t \leq T_k, C_k + \frac{(\sum_{\tau_i \in {\bf T}'} C_i) +  (\sum_{i=1}^{k-1}W_i^{normal}(t)) }{M} \leq t.  
\end{equation}
for all ${\bf T}' \subseteq hp(\tau_k)$ with $|{\bf T}'| = M-1$.

This leads to the following theorem by using
Lemma~\ref{lemma:framework-general-schedulability}. 

\begin{theorem}
\label{thm:multiprocessor-GRM-M-1-carry-k2q}
Task $\tau_k$ in a sporadic implicit-deadline task system is
schedulable by global RM on $M$ processors if $\sum_{i=1}^{k-1}
  C_i \leq M T_k$ and
\begin{equation}\small
\label{eq:schedulability-GRM-M-1-carry-k2q}
 U_k \leq 1 -   \frac{\sum_{\tau_i \in {\bf
      T}'} C_i}{M T_k} -  \sum_{i=1}^{k-1} \frac{U_i}{M} -
\frac{\sum_{i=1}^{k-1} C_i}{M T_k}+\frac{\sum_{i=1}^{k-1} ( U_i
  \sum_{\ell=i}^{k-1} C_\ell )}{M^2 T_k}.
\end{equation}
by indexing the $k-1$ higher-priority tasks in a non-decreasing order
of $(\ceiling{\frac{T_k}{T_i}}-1)T_i$ for every $\tau_i \in hp(\tau_k)$ and by putting the $M-1$
higher-priority tasks with the largest execution times into ${\bf
  T}'$.
\end{theorem}
\begin{proof}
  It is not necessary to enumerate all ${\bf T}'\subseteq {\bf T}$
  with $|{\bf T}'| = M-1$ if we can construct the task set ${\bf T}'
  \subseteq hp(\tau_k)$ with the maximum $\sum_{\tau_i \in {\bf T}'}
  C_i$.  To use \frameworkkq{}, we are certain about which tasks
  should be put into the carry-in task set ${\bf T}'$ by assuming that
  $C_i$ and $T_i$ are both given. That is, we simply have to put the
  $M-1$ higher-priority tasks with the largest execution times into
  ${\bf T}'$. This can be imagined as if we increase the execution
  time of task $\tau_k$ from $C_k$ to $C_k' = C_k + \frac{\sum_{\tau_i
      \in {\bf T}'} C_i}{M}$. Moreover, we have $\alpha_i =
  \frac{1}{M}$ and $\beta_i = \frac{1}{M}$ for every task $\tau_i \in
  hp(\tau_k)$ in this case. 

  Therefore, based on the test in
  Eq.~\eqref{eq:grm-multiprocessor-M-1-carryin}, we have the 
  last release time ordering defined by indexing the $k-1$
  higher-priority tasks in a non-decreasing order of
  $(\ceiling{\frac{T_k}{T_i}}-1)T_i$ for every $\tau_i \in
  hp(\tau_k)$. By adopting
  Lemma~\ref{lemma:framework-general-schedulability} with
  $\alpha_i=\frac{1}{M}$ and $\beta_i = \frac{1}{M}$, we know that
  task $\tau_k$ is schedulable by global RM if $\sum_{i=1}^{k-1}
  \frac{C_i}{M} \leq T_k$ and 
\begin{equation}\small
 \frac{C_k + \sum_{\tau_i \in {\bf
      T}'} \frac{C_i}{M}}{T_k} \leq 1 -  \sum_{i=1}^{k-1} \frac{U_i}{M} -
\frac{\sum_{i=1}^{k-1} C_i}{M T_k}+\frac{\sum_{i=1}^{k-1} ( U_i
  \sum_{\ell=i}^{k-1} C_\ell )}{M^2 T_k}.
\end{equation}
  By reorganizing the above inequality, we reach the conclusion.
\end{proof}

We can always take the pessimistic last release time ordering in
Lemma~\ref{lemma:general-sorting}, for concluding the following theorem.

\begin{theorem}
\label{thm:multiprocessor-GRM-M-1-carry-k2q-V2}
Task $\tau_k$ in a sporadic implicit-deadline task system is
schedulable by global RM on $M$ processors if the condition in
Eq.~\eqref{eq:schedulability-GRM-M-1-carry-k2q} holds by indexing the
$k-1$ higher-priority tasks in a non-increasing order of $T_i$, for
every $\tau_i \in hp(\tau_k)$.
\end{theorem}
\begin{proof}
  This is proved based on the same argument in Theorem~\ref{thm:multiprocessor-GRM-M-1-carry-k2q} by adopting Lemmas~\ref{lemma:framework-general-schedulability}~and~\ref{lemma:general-sorting}.
\end{proof}

We can of course revise the statement in
Theorems~\ref{thm:multiprocessor-GRM-M-1-carry-k2q}~and~\ref{thm:multiprocessor-GRM-M-1-carry-k2q-V2}
by adopting Lemma~\ref{lemma:framework-constrained-schedulability} and
Lemma~\ref{lemma:framework-totalU-exclusive} to construct
schedulability tests by using only the utilization of the higher-priority tasks.

{\bf Evaluation Results} We conduct experiments using synthesized task
sets for evaluating the tests in
Theorem~\ref{thm:multiprocessor-GRM-M-1-carry-k2q} and
Theorem~\ref{thm:multiprocessor-GRM-M-1-carry-k2q-V2}.  We first
generated a set of sporadic tasks. The cardinality of the task set was
$5$ times the number of processors, i.e., 40 tasks on 8 multiprocessor
systems.  The UUniFast-Discard method~\cite{davis2011improved} was
adopted to generate a set of utilization values with the given goal.
We used the approach suggested by Davis et
al.~\cite{davis2008efficient} to generate the task periods according
to a uniform distribution in the range of the logarithm of the task periods (i.e.,
log-uniform distribution).  The order of magnitude $p$ to
control the period values between the largest and smallest periods is
parameterized in evaluations, (e.g., $1-10ms$ for $p=1$, $1-100ms$ for
$p=2$, etc.).  We evaluate these tests in uniprocessor systems with $p
\in [1,2,3]$.  The execution time
was set accordingly, i.e., $C_{i}=T_iU_i$.  Tasks' relative deadlines
were equal to their periods.

The evaluated tests for $n$ tasks in ${\bf T}$ with $n \geq M$ are:
\begin{itemize}
\item \emph{BCL}: the linear-time test in Theorem 4 in~\cite{bertogna2006new}.
\item \emph{FF}: the pseudo-polynomial-time forced-forward (FF) analysis in Eq. (5) in \cite{DBLP:journals/rts/BaruahBMS10}.
\item \emph{BAK}: the $O(n^3)$ test in Theorem 11 in~\cite{baker2006analysis}.
\item \emph{Guan}: the pseudo-polynomial-time response time analysis\cite{DBLP:conf/rtss/GuanSYY09}.
\item \emph{QB-BC} (from \frameworkkq{}):
  Eq.~\eqref{eq:schedulability-GRM-M-1-carry-k2q} in
  Theorem~\ref{thm:multiprocessor-GRM-M-1-carry-k2q}. This requires to
  sort the higher-priority tasks to define the proper last release
  ordering and the $M-1$ carry-in jobs; therefore, the time complexity is $O(n^2 \log n)$ for a
  task set with $n$ tasks.
\item \emph{QB-BC2} (from \frameworkkq{}):
  Eq.~\eqref{eq:schedulability-GRM-M-1-carry-k2q} in
  Theorem~\ref{thm:multiprocessor-GRM-M-1-carry-k2q-V2} by always
  using the worst-case release time ordering, which is the reverse
  order of the given priority assignment. The schedulability test can
  be implemented in $O(n\log M)$ time complexity by using proper data
  structures, provided that the RM priority order is given.\footnote{The time complexity is mainly due to the
    calculation of ${\bf T}'$ to get the $M-1$ tasks with the maximum
    carry-in execution time since the other operations can be done in
    $O(1)$ time complexity by using proper data structures to
    calculate the values when we intend to test task $\tau_{k+1}$
    after task $\tau_k$. Specifically, due to the predefined last
    release time ordering, when we intend to test task $\tau_{k+1}$
    after task $\tau_k$, we only have to insert task $\tau_k$ to be
    indexed as $1$ and updating from $\frac{\sum_{i=1}^{k-1} ( U_i
      \sum_{\ell=i}^{k-1} C_\ell )}{M^2 T_k}$  to $\frac{\sum_{i=1}^{k} ( U_i
      \sum_{\ell=i}^{k} C_\ell )}{M^2 T_k}$ (under the new ordering) takes only constant time
    complexity.  Finding task set ${\bf T}'$ can be implemented by
    using a min heap to store the $M-1$ tasks in ${\bf T}'$. When we move
    from testing task $\tau_k$ (when $k \geq M$) to task $\tau_{k+1}$,
    we need to compare whether $C_k$ is larger than the minimum
    execution time of the tasks in the heap. If no, we keep the same
    task set ${\bf T}'$; if yes, we pop out the task with the minimum
    execution time in the heap, and insert task $\tau_k$ into the
    heap.  By using the heap, this operation requires time complexity
    $O(\log M)$. Calculating $C_{k+1}'$ from $C_k$ with
    the help of the heap can be done in $O(1)$ time complexity. }
\end{itemize}

Figure~\ref{fig:mul-2} depicts the result of the performance
comparison. In all the cases, we can see that QB-BC is superior to all
the other polynomial-time tests.  QB-BC2 is slightly worse than QB-BC
but the time complexity is lower.  Since QB-BC and QB-BC2 are designed
from a more pessimistic test than the analysis by Guan et
al. \cite{DBLP:conf/rtss/GuanSYY09} in pseudo-polynomial time, they 
are worse. But, we note that
there is a significant gap in time complexity between QB-BC, QB-BC2,
and Guan.  Overall, the tests derived by using the \frameworkkq{}
framework perform reasonably well with their low time complexity.

\begin{figure*}[t]
  \centering
  \includegraphics[width=0.9\textwidth]{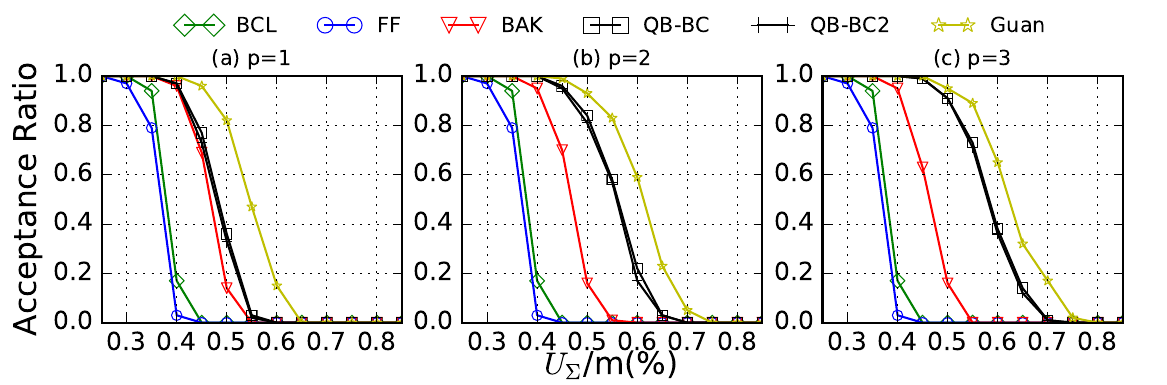}
  \caption{Acceptance ratio comparison on implicit-deadline 8 multiprocessor systems.}
  \label{fig:mul-2}
\end{figure*}

\section{Conclusion and Extensions}

In this paper, we present a general response-time analysis and
schedulability-test framework, called \frameworkkq{}. Thanks to the independence
upon the task and platform models in the framework, \frameworkkq{} can be viewed as a
``block-box'' interface that can result in sufficient
utilization-based analyses for a wide range of applications in
real-time systems under fixed-priority scheduling.  We believe that
the \frameworkkq{} framework has high potential to be adopted to solve
several other problems for analyzing other task models in real-time
systems with fixed-priority scheduling.  The framework can be used,
once the corresponding $k$-point last-release scheduling test or
response time analysis can be constructed.

Moreover, our proposed frameworks,
\frameworkku{} and \frameworkkq{},
provide a solid mathematical foundation for deriving polynomial-time utilization-based
schedulability tests and response time analyses almost
\emph{automatically}. That is, utilization-based analyses are almost
automatically derived if the schedulability tests can be formulated
in the scope of the frameworks.
\ifbool{techreport}{
  We have demonstrated several applications in this paper. Some models
  have introduced pretty high dynamics, but we can still handle the
  response time analysis and schedulability test with proper
  constructions so that the \frameworkkq{} framework is applicable.
Therefore, with the presented approach, some difficult schedulability
test and response time analysis problems may be solved by building a good (or
exact) exponential-time test and using the approximation in the
\frameworkkq{} framework.  With the quadratic and hyperbolic
expressions, \frameworkkq{} and \frameworkku{} frameworks can be used
to provide many quantitive features to be measured, like the total
utilization bounds, speed-up factors, etc., not only for uniprocessor
scheduling but also for multiprocessor scheduling.
}{
  Due to the space limitation, we are only able to summarize some results
  herein, but not to provide details. The detailed evaluations,
  compared to other approaches are in
  \cite{DBLP:journals/corr/framework-compare}. There are more
  applications presented in the report
  \cite{DBLP:journals/corr/abs-k2q}, including
\begin{compactitem}
\item more explorations for uniprocessor fixed-priority scheduling in
  Appendix B\citetechreport{}.  
\item an extension to adopt the force-forward schedulability analysis for multiprocessor global fixed-priority scheduling in
  Appendix C\citetechreport{}.  
\item the first polynomial-time worst-case response time analysis, to
  the best of our knowledge, for sporadic real-time tasks with jitters
  Appendix D\citetechreport{}.  
\item a demonstration on converting
  exponential-time schedulability tests of generalized multi-frame
  task models
  \cite{DBLP:journals/rts/BaruahCGM99,DBLP:conf/rtcsa/TakadaS97} to
  polynomial-time tests in 
  Appendix E\citetechreport{}.
\item mode-level fixed-priority scheduling
  policies by
  studying the acyclic task model
  \cite{DBLP:journals/tc/AbdelzaherSL04} and the multi-mode task model
  \cite{DBLP:conf/rtas/DavisFPS14} in Appendix F\citetechreport{}.
\end{compactitem}
}

When adopting \frameworkkq{} for schedulability tests, we assume that
$t_k$ is specified in
Lemma~\ref{lemma:framework-general-schedulability}. In this paper, we
do not explore how to configure the best value of $t_k$ and its last
release time ordering $\pi$ such that the resulting quadratic form is
the best. Therefore, the combination of \frameworkkq{}/\frameworkku{}
and the tunable approach by Bini and Buttazzo
\cite{DBLP:journals/tc/BiniB04} can be an interesting future research
direction, as this can potentially balance the schedulability test and
the time complexity for concrete applications. Essentially, this
combination is to search the proper settings of different $t_k$ values
such that the associated last release time ordering $\pi$ can be 
less pessimistic, as demonstrated by several cases regarding Example~\ref{example-3tasks} in Section~\ref{sec:framework}.

 \begin{spacing}{0.9}
   \noindent{\small {\bf Acknowledgement}: This paper has been
     supported by NSF grant CNS 1527727 and DFG, as part of the Collaborative Research Center
     SFB876 (http://sfb876.tu-dortmund.de/), and the priority program
     "Dependable Embedded Systems" (SPP 1500 -
     http://spp1500.itec.kit.edu). \ifbool{techreport}{}{We would also like to thank the
     anonymous reviewers for their valuable comments to improve the
     presentation of the paper.}
  }
 \end{spacing}

\footnotesize
\vspace{-0.1in}
\begin{spacing}{0.98}
\def\IEEEbibitemsep{-1pt}
\bibliographystyle{abbrv}
\bibliography{ref,real-time}
\end{spacing}
\normalsize

\ifbool{techreport}{}{\vspace{-0.2in}}
\section*{Appendix A: Proofs}

 

  

\begin{appProof}{Lemma \ref{lemma:framework-totalU-constrained}}
 Similar to the proof of Lemma~\ref{lemma:framework-totalU-exclusive}, we only have to consider the cases when $U_i$ is set to $\frac{x}{k-1}$ to make the schedulability condition the most difficult, where $x=\sum_{i=1}^{k-1}U_i$. Suppose that $\frac{C_k}{t_k}$ is $y$. Then, we are looking for the infimum $x+y$ such that $y > 1- (\alpha+\beta)x+0.5\alpha\beta (x^2 + \frac{x^2}{k-1})$.

  To solve this, we start with $y = 1- (\alpha+\beta)x+0.5\alpha\beta (x^2 + \frac{x^2}{k-1})$. Our objective becomes to minimize $H(x)=x+1- (\alpha+\beta)x+0.5\alpha\beta (x^2 + \frac{x^2}{k-1})$. By finding $\dfrac{d H(x)}{d x} = 1-(\alpha+\beta)+\frac{k\alpha\beta x}{k-1}=0$, we know that $x=\frac{(k-1)(\alpha+\beta-1)}{k \alpha\beta}$. Therefore,
  \begin{align}
    y =& 1 - x\left(\alpha+\beta - \frac{0.5 k \alpha\beta \frac{(k-1)(\alpha+\beta-1)}{k \alpha\beta}}{k-1}\right) \nonumber\\
   = & 1-\frac{(k-1)(\alpha+\beta-1)}{k \alpha\beta}(0.5(\alpha+\beta+1)) \nonumber\\
   = & 1-\frac{0.5(k-1)\left((\alpha+\beta)^2-1\right)}{k\alpha\beta}.\label{eq:y-and-k}
  \end{align}
Since $\alpha+\beta \geq 1$, we know that $x \geq 0$.
  Whether we should take the above solution only depends on whether $y \geq 0$ or not. If $y \geq 0$, then we can conclude the solution directly; otherwise, if $y < 0$, we should set $y$ to $0$. That is, by reorganizing Eq.~\eqref{eq:y-and-k} (under the assumption $\alpha> 0$ and $\beta > 0$), examining whether $y < 0$ is equivalent to testing
$(1-\frac{1}{k})\left((\alpha+\beta)^2-1\right) > 2\alpha\beta$, which implies to test whether $\alpha^2+\beta^2-1> \frac{1}{k}\left((\alpha+\beta)^2-1\right)$. If $\alpha^2+\beta^2 \leq 1$, then  $y \geq 0$ since $\alpha^2+\beta^2-1 \leq 0 \leq \frac{1}{k}\left((\alpha+\beta)^2-1\right)$ due to the assumption $\alpha+\beta \geq 1$.
Therefore, there are two cases:\\
\noindent {\bf Case 1:} If $\alpha^2+\beta^2 > 1$ and $k >
\frac{(\alpha+\beta)^2-1}{\alpha^2+\beta^2-1}$, then, for such a case $y$ derived from Eq.~\eqref{eq:y-and-k} is negative. We should set $y$ to $0$. The remaining procedure here is the same as in solving the quadratic equation in the proof of Lemma~\ref{lemma:framework-totalU-exclusive} by setting $\frac{C_k}{t_k}$ to $0$. This leads to the first condition in Eq.~\eqref{eq:schedulability-totalU-constrained}.\\
\noindent{\bf Case 2:} If  $\alpha^2+\beta^2 \leq1$ or $k \leq
\frac{(\alpha+\beta)^2-1}{\alpha^2+\beta^2-1}$,
then, we have the
conclusion that $y \geq 0$ and $x \geq 0$. We just have to sum up the above derived $x$
and $y$.  This leads to the second condition in
Eq.~\eqref{eq:schedulability-totalU-constrained} directly.
\end{appProof}

\begin{appProof}{Lemma \ref{lemma:framework-general-response}}
 Definition \ref{def:kpoints-response} leads to the following optimization problem:

  \begin{subequations}\label{eq:lp-init-response0}
{\footnotesize  \begin{align}
    \mbox{sup\;\;} & C_k + \sum_{i=1}^{k-1} \alpha_i t_i^* U_i + \sum_{i=1}^{k-1} \beta_i C_i \\
    \mbox{such that\;\;} &     C_k + \sum_{i=1}^{k-1} \alpha_i t_i^*
    U_i + \sum_{i=1}^{j-1} \beta_i C_i > t_j^* ,&\forall j=1,\ldots,
    k-1,  \\
&t_j^* \geq 0,&\forall j=1,\ldots, k-1,     
  \end{align}    }
  \end{subequations}
\noindent where $t^*_1, t^*_2, \ldots, t^*_{k-1}$ and are variables,
  $\alpha_i$, $\beta_i$, $U_i$, $C_i$ for higher-priority task
  $\tau_i$ and $C_k$ are constants.
  For the rest of the proof, we replace $>$ with $\geq$ in
  Eq.~(\ref{eq:lp-init-response0}), as the supermum
  and the maximum are the same when presenting the inequality with
  $\geq$.  We can also further drop the condition $t_j^* \geq 0$, which just makes the resulting solution more pessimistic. This
  results in the following linear programming, which has a safe upper bound of
  Eq.~\eqref{eq:lp-init-response0},

  \begin{subequations}\label{eq:lp-init-response}
{\footnotesize  \begin{align}
    \mbox{maximize\;\;} & C_k + \sum_{i=1}^{k-1} \alpha_i t_i^* U_i + \sum_{i=1}^{k-1} \beta_i C_i \\
    \mbox{such that\;\;} &     C_k + \sum_{i=1}^{k-1} \alpha_i t_i^*
    U_i + \sum_{i=1}^{j-1} \beta_i C_i \geq t_j^* ,\forall j=1,\ldots,
    k-1.    \label{eq:lp-constraint-response}
  \end{align}    }
  \end{subequations}

  The linear programming in Eq.~(\ref{eq:lp-init-response}) (by
  replacing $>$ with $\geq$ and supremum with maximum) has $k-1$
  variables and $k-1$ constraints. Like the proof of Lemma~\ref{lemma:framework-general-schedulability}, we again adopt the extreme point
  theorem for linear programming \cite{luenberger2008linear} to solve
  the linear programming. Suppose that $t_1^\dagger, t_2^\dagger, \ldots,
  t_{k-1}^\dagger$ is a feasible solution for the linear programming in
  \eqref{eq:lp-init-response} and $\hat{t} = \max\{t_1^\dagger, t_2^\dagger, \ldots,
  t_{k-1}^\dagger\}$. By the satisfaction of
  Eq.~\eqref{eq:lp-constraint-response}, we know that
  \begin{align*}\small
    C_k + \hat{t} \sum_{i=1}^{k-1} \alpha_i U_i + \sum_{i=1}^{k-1} \beta_i C_i
   \geq C_k + \sum_{i=1}^{k-1} \alpha_i t_i^\dagger U_i +
   \sum_{i=1}^{k-1} \beta_i C_i \geq \hat{t}.
  \end{align*}
  As a result, we have $\hat{t} \leq \frac{C_k +
    \sum_{i=1}^{k-1}\beta_i C_i}{1-\sum_{i=1}^{k-1} \alpha_i
    U_i}$. That is, any feasible solution of
  Eq.~\eqref{eq:lp-init-response} has $t_j^* \leq \frac{C_k +
    \sum_{i=1}^{k-1}\beta_i C_i}{1-\sum_{i=1}^{k-1} \alpha_i U_i}$ for
  any $j=1,2,\ldots,k-1$. Under the assumption that $\sum_{i=1}^{k-1}
  \alpha_i U_i < 1$ and $0 \leq \sum_{i=1}^{k-1} \beta_i C_i$, 
 the  above linear programming has a bounded objective function.

  The only extreme point solution is to put $C_k + \sum_{i=1}^{k-1}
  \alpha_i t_i^* U_i + \sum_{i=1}^{j-1} \beta_i C_i = t_j^*$ for every
  $j=1,2,\ldots,k-1$. Since the objective function is bounded, by the
  extreme point theorem \cite{luenberger2008linear}, we know
  that this extreme point solution is the optimal solution for the
  linear programming in Eq.~\eqref{eq:lp-init-response}. For such a solution, we know that
\begin{equation}
\label{eq:tj-response-time}
\forall j=2,3,\ldots,k-1, \qquad t_j^* - t_{j-1}^*= \beta_{j-1} C_{j-1}.
\end{equation}
and 
\begin{equation}
\label{eq:t1-response-time}  
t_1^* = C_k + \sum_{i=1}^{k-1} \alpha_i U_i (t_1^* + \sum_{\ell=0}^{i-1} \beta_{\ell} C_{\ell}),
\end{equation} 
where $\beta_0$ and $C_0$ are defined as $0$ for notational brevity. 
Therefore, we know that
\begin{equation}
\label{eq:t1-response-time-final} 
t_1^* = \frac{C_k+\sum_{i=1}^{k-1} \alpha_i U_i (\sum_{\ell=0}^{i-1} \beta_{\ell} C_{\ell})}{1-\sum_{i=1}^{k-1} \alpha_i U_i}.
\end{equation}
Clearly, the above extreme point solution is always feasible when
$\sum_{i=1}^{k-1} \alpha_i U_i < 1$.
Therefore,  in this extreme point solution, the objective function of
the linear programming is

{\footnotesize  \begin{align}
    \label{eq:fina-response-lemma}
t_1^* + \sum_{i=1}^{k-1} \beta_i C_i =& \frac{C_k+\sum_{i=1}^{k-1} \alpha_i U_i (\sum_{\ell=0}^{i-1} \beta_{\ell} C_{\ell})}{1-\sum_{i=1}^{k-1} \alpha_i U_i} + \sum_{i=1}^{k-1} \beta_i C_i \\
= & \frac{C_k+ \sum_{i=1}^{k-1} \beta_i C_i - \sum_{i=1}^{k-1} \alpha_i U_i (\sum_{\ell=i}^{k-1} \beta_{\ell} C_{\ell})}{1-\sum_{i=1}^{k-1} \alpha_i U_i}
  \end{align}}
  which concludes the proof.  
\end{appProof}

\ifbool{techreport}{}{\end{document}}

\section*{Appendix B: Quadratic Bound for Uniprocessor
  Constrained-Deadline Tasks}

To verify the schedulability of a
(constrained-deadline) sporadic real-time task $\tau_k$ under fixed-priority scheduling
in uniprocessor systems, the time-demand analysis (TDA) developed in
\cite{DBLP:conf/rtss/LehoczkySD89} can be adopted.
 That is, if
\begin{equation}
  \label{eq:exact-test-constrained-deadline}
\exists t \mbox{ with } 0 < t \leq D_k {\;\; and \;\;} C_k +
\sum_{\tau_i \in hp(\tau_k)} \ceiling{\frac{t}{T_i}}C_i \leq t,
\end{equation}
then task $\tau_k$ is schedulable under the fixed-priority scheduling algorithm, where $hp(\tau_k)$ is the set of tasks with higher priority than $\tau_k$, $D_k$, $C_k$, and $T_i$ represent $\tau_k$'s relative deadline, worst-case execution time, and period, respectively. 
For a constrained-deadline task $\tau_k$, the schedulability test in
Eq.~\eqref{eq:exact-test-constrained-deadline}  is equivalent to the
verification of the existence of $0 < t \leq D_k$ such that
\begin{equation}
  \label{eq:exact-test-constrained-deadline-2}
 C_k + \sum_{\tau_i \in hp_2(\tau_k)} C_i + \sum_{\tau_i \in hp_1(\tau_k)} \ceiling{\frac{t}{T_i}}C_i \leq t.  
\end{equation}
We can then create a virtual sporadic task $\tau_k'$ with execution
time $C_k'=C_k + \sum_{\tau_i \in hp_2(\tau_k)} C_i$, relative
deadline $D_k'=D_k$, and period $T_k'=D_k$. It is
clear that the schedulability test to verify the schedulability of
task $\tau_k'$ under the interference of the higher-priority tasks
$hp_1(\tau_k)$ is the same as that of task $\tau_k$ under the
interference of the higher-priority tasks $hp(\tau_k)$.
For notational brevity, suppose that there are $k^*-1$ tasks in
$hp_1(\tau_k)$.

\begin{theorem}
\label{theorem:sporadic-general-v2}
Task $\tau_k$ in a sporadic task system with constrained deadlines is
schedulable by the fixed-priority scheduling algorithm if
$\sum_{i=1}^{k^*-1}\frac{C_i}{D_k} \leq 1$ and
\begin{equation}
\label{eq:schedulability-sporadic-any-a}
\frac{C_k'}{D_k} \leq 1-\sum_{i=1}^{k^*-1}U_i-\sum_{i=1}^{k^*-1}\frac{C_i}{D_k}+\frac{\sum_{i=1}^{k^*-1} U_i (\sum_{\ell=i}^{k^*-1}  C_\ell)}{D_k},
\end{equation}
in which the $k^*-1$ higher-priority tasks in $hp_1(\tau_k)$ are indexed
in a non-decreasing order of $\left(\ceiling{\frac{D_k}{T_i}}-1\right)T_i$.
\end{theorem}
\begin{proof}
Setting $t_i =\left (\ceiling{\frac{D_k}{T_i}}-1\right)T_i$, and indexing the
tasks in a non-decreasing order of $t_i$ leads to the satisfaction of
Definition~\ref{def:kpoints} with $\alpha_i=1$ and $\beta_i=1$.
\end{proof}

\begin{corollary}
\label{corollary-rm}
Task $\tau_k$ in a sporadic task system with implicit deadlines is
schedulable by the RM scheduling algorithm if
Lemmas~\ref{lemma:framework-general-schedulability},~\ref{lemma:framework-constrained-schedulability},
\ref{lemma:framework-totalU-exclusive},
or~\ref{lemma:framework-totalU-constrained} holds by setting
$\frac{C_k}{t_k}$ as $U_k$, $\alpha=1$, and $\beta=1$.
\end{corollary}

The above result in Corollary~\ref{corollary-rm} leads to the
utilization bound $2-\sqrt{2}$ (by using
Lemma~\ref{lemma:framework-totalU-constrained} with $\alpha=1$ and
$\beta=1$) for RM scheduling, which is worse than the existing
Liu and Layland bound $\ln{2}$ \cite{liu1973scheduling}. 

\section*{Appendix C: Multiprocessor DM/RM Scheduling}

This part demonstrates how to use the \frameworkkq{} framework for multiprocessor
global fixed-priority scheduling. We consider that the system has $M$
identical processors. For global fixed-priority scheduling, there is a
global queue and a global scheduler to dispatch jobs. We
demonstrate the applicability for constrained-deadline and
implicit-deadline sporadic systems under global fixed-priority
scheduling. Specifically, we will present how to apply the
framework to obtain speed-up and capacity augmentation factors for
global DM and global RM.

The success of the scheme depends on a corresponding
exponential-time test. Here we will use the property to be presented
in Lemma~\ref{lemma:multiprocessor-grm-sporadic-pushforward}, based on
the \emph{forced-forward} algorithm proposed by Baruah et
al. \cite{DBLP:journals/rts/BaruahBMS10} to characterize the workload
of higher-priority tasks. The method in
\cite{DBLP:journals/rts/BaruahBMS10} to analyze fixed-priority
scheduling is completely different from ours, as they rely on the
demand bound functions of the tasks. 

The following
lemma provides a sufficient test based on the observations by Baruah
et al. \cite{DBLP:journals/rts/BaruahBMS10}.
The construction of the following lemma is based 
on a minor change of the forced-forward algorithm.

\begin{lemma}
\label{lemma:multiprocessor-grm-sporadic-pushforward} 
Let $\Delta_k^{\max}$ be $\max_{j=1}^{k-1} \{U_j, \frac{C_k}{D_k}\}$.   
Task $\tau_k$ in a sporadic task system with constrained deadlines  is
schedulable by a global fixed-priority (workload conserving)
scheduling algorithm on
$M$ processors if
\begin{align*}
  \forall y \geq 0,  \left(\forall 0 \leq \omega_i \leq T_i, \forall \tau_i
  \in hp(\tau_k)\right), \exists t \mbox{ with } 0 < t \leq D_k+y\\  \mbox{ such that \;\;} 
  \Delta_k^{\max} \cdot(D_k+y) +
 \frac{ \sum_{i=1}^{k-1} \omega_i\cdot U_i+ \ceiling{\frac{t-\omega_i }{T_i}} C_i
}{M}\leq t.
\end{align*}
\end{lemma}
\begin{proof}
  This is proved by contrapositive. If task $\tau_k$ is not
  schedulable by the global fixed-priority scheduling, we will show that there exist $y \geq 0$ and $0 \leq \omega_i \leq T_i$ such that for all $0 < t \leq D_k+y$, the condition $\Delta_k^{\max} \cdot (D_k+y) +
 \frac{ \sum_{i=1}^{k-1} \omega_i\cdot U_i+ \ceiling{\frac{t-\omega_i }{T_i}} C_i
}{M} > t$ holds. The proof is mainly based on the forced-forward algorithm for the analysis of global DM by Baruah et al. in \cite{DBLP:journals/rts/BaruahBMS10}, by making some further annotations.

  If $\tau_k$ is not schedulable by global DM, let $z_0$ be the first time at which
  task $\tau_k$ misses its absolute deadline, i.e., $z_0$.  Let $z_1$
  be the arrival time of this job of task $\tau_k$.  For notational
  brevity, let this job be $J_1$, which arrives at time $z_1$ and has
  not yet been finished at time $z_0$. By definition, we know that
  $z_0-z_1$ is $D_k$. Due to the fixed-priority and
  workload-conserving scheduling policy and the constrained-deadline setting, removing (1) all the other jobs of task $\tau_k$ (except the one arriving at time $z_1$), (2) all the jobs arriving no earlier than $z_0$,  and (3) lower-priority jobs does not change the unschedulability of job $J_1$. Therefore, the rest of the proof only considers the jobs from $\tau_1, \tau_2, \ldots, \tau_k$.

  Now, we expand the window of interest by using a slightly
  different algorithm from that proposed in
  \cite{DBLP:journals/rts/BaruahBMS10}, also illustrated with the
  notation in Figure~\ref{fig:force-forward}, as in Algorithm~\ref{alg:feasibility-analysis2}. The difference is only in the setting
``strictly less than $(z_{\ell-1}-z_\ell)\cdot \hat{U}_\ell$ units'', whereas the setting in \cite{DBLP:journals/rts/BaruahBMS10} uses 
``strictly less than $(z_{\ell-1}-z_\ell)\cdot s$ units'' for a
certain $s$.
For notationaly brevity, $\hat{U}_\ell$ is the utilization of the task that generates
job $J_\ell$.

 \begin{algorithm}[h]
   \caption{(Revised) Forced-Forward Algorithm}
   \label{alg:feasibility-analysis2}
    \begin{algorithmic}[1]\footnotesize
        \FOR {$\ell\leftarrow2,3,...$}
          \STATE let $J_\ell$ denote a job that 
          \begin{compactitem}[-]
          \item arrives at some time-instant $z_\ell<z_{\ell-1}$;
          \item has an absolute deadline after $z_{\ell-1}$;
          \item has not completed execution by $z_{\ell-1}$; and
          \item has executed for strictly less than $(z_{\ell-1}-z_\ell)\cdot \hat{U}_\ell$
	 units over the interval $[z_\ell,z_{\ell-1})$, where $\hat{U}_\ell$ is the utilization of the task that generates job $J_\ell$.
          \end{compactitem}

          \IF {there is no such a job}
          \STATE $\ell\leftarrow (\ell-1)$; break;
          \ENDIF
      \ENDFOR
     \end{algorithmic}
     \end{algorithm}

     Suppose that the forced-forward algorithm terminates with $\ell$ equals to $\ell^*$. We now examine the schedule in the interval $(z_{\ell^*}, z_0]$. Since $J_1, J_2, \ldots, J_{\ell^*}$ belong to $\tau_1, \tau_2, \ldots, \tau_k$, we know that $\hat{U}_\ell \leq \Delta_k^{\max}$ for $\ell=1,2,\ldots,\ell^*$. Let $\sigma_\ell$ be the total length of the time interval over $(z_\ell, z_{\ell-1}]$ during which $J_\ell$ is executed. By the choice of $J_\ell$, it follows that
\ifbool{techreport}{\[}{$}
     \sigma_\ell < (z_{\ell-1}-z_{\ell}) \cdot \hat{U}_\ell \leq (z_{\ell-1}-z_{\ell}) \cdot \Delta_k^{\max}.
\ifbool{techreport}{\]}{$}
     Moreover, all the $M$ processors execute other higher-priority jobs (than $J_\ell$) at any other time points in the interval $(z_\ell, z_{\ell-1}]$ at which $J_\ell$ is not executed.
     Therefore, we know that the maximum amount of time from $z_{\ell^*}$ to $z_0$, in which not all the $M$ processors execute certain jobs, is at most
     \begin{align*}\small
       \sum_{\ell=1}^{\ell^*} \sigma_\ell &<      
     \sum_{\ell=1}^{\ell^*} (z_{\ell-1}-z_{\ell}) \cdot \Delta_k^{\max} = (z_0-z_{\ell^*}) \cdot \Delta_k^{\max}.       
   \end{align*}

  Up to here, the treatment is almost identical to that in ``Observation 1'' in \cite{DBLP:journals/rts/BaruahBMS10}. The following analysis becomes different as we do not intend to use the demand bound function.
   Now, we replace job $J_1$ with another job $J_1'$ with \emph{inflated execution time} in the above schedule, where $J_1'$ is released at time $z_{\ell^*}$ with absolute deadline $z_0$ and execution time $(z_0-z_{\ell^*}) \cdot \Delta_k^{\max}$. According to the above analysis, $J_1'$ cannot be finished before $z_{0}$ in the above schedule.
   For each task $\tau_i$ for $i=1,2,\ldots,k-1$, in the above schedule, there may be one \emph{carry-in} job, denoted as $J_i^\flat$, of $\tau_i$ (under the assumption of the schedulability of a higher-priority task $\tau_i$) that arrives at time $r_i^\flat$ with $r_i^\flat < z_{\ell^*}$ and $r_i^\flat+T_i > z_{\ell^*}$. Let $d_i^\flat$ be the next released time of task $\tau_i$ after $z_{\ell^*}$, i.e., $d_i^\flat = r_i^\flat + T_i$.

According to the termination condition in the construction of $z_{\ell^*}$, if $J_i^\flat$ exists, we know that at least $(z_{\ell^*}-r_i^\flat)\cdot U_i$ amount of execution time has been executed before $z_{\ell^*}$, and the remaining execution time of job $J_i^\flat$ to be executed after $z_{\ell^*}$ is at most
   $(d_i^\flat - z_{\ell^*})\cdot U_i$. If $J_i^\flat$ does not exist, then $d_i^\flat$ is set to $z_{\ell^*}$ for notational brevity. Therefore, the amount of workload $W_i'(t)$ of all the released jobs of task $\tau_i$ for $i=1,2,\ldots,k-1$ to be executed in time interval $(z_\ell^*, z_\ell^*+t)$ is at most
   \begin{equation}
     \label{eq:1}
   W_i'(t)  = (d_i^\flat - z_{\ell^*})\cdot U_i+ \ceiling{\frac{t-(d_i^\flat - z_{\ell^*}) }{T_i}} C_i.     
   \end{equation}
   
   The assumption of the unschedulability of job $J_1'$ (due to the
   unschedulability of job $J_1$) under the global fixed-priority
   scheduling implies that $J_1'$ cannot finish its computation at any time between $z_{\ell^*}$ and $z_0$. This leads to the following (necessary) condition 
  $\Delta_k^{\max} (z_0-z_{\ell^*}) +
  \frac{ \sum_{\tau_i \in hp(\tau_k)}W_i'(t) }{M} >  t$ for all $0 < t \leq z_{\ell^*}-z_0$ for the unschedulability of job $J_1'$.   

  Therefore, by the existence of $y=z_{1}-z_{\ell^*}$ (with $y \geq 0$) and $\omega_i = d_i^\flat - z_{\ell^*}$  (with $0 \leq \omega_i \leq T_i$) for $i=1,2,\ldots, k-1$ to enforce the above necessary condition, we reach the conclusion of the proof by contrapositive.  That is, task $\tau_k$ is schedulable if, for all $y \geq 0$ and any combination of $0 \leq \omega_i \leq T_i$ for $i=1,2,\ldots, k-1$, there exists $0 < t \leq D_k+y$ with $\Delta_k^{\max} (D_k+y) +
  \frac{ \sum_{\tau_i \in hp(\tau_k)} \omega_i\cdot U_i+ \ceiling{\frac{t-\omega_i}{T_i}} C_i }{M}  \leq t$.
\end{proof}

The schedulability condition in
Lemma~\ref{lemma:multiprocessor-grm-sporadic-pushforward} may look
at the first glance strange. We briefly explain the logical meaning (but informally) here.  Suppose we
would like to know whether a job of task $\tau_k$ arrived at time
$r_k$ can be finished before/at time $r_k+D_k$. To better quantify the
interference from the higher-priority tasks, we would like to account for the higher-priority jobs arrived prior to $r_k$.  The
variable $y$ defines the extension of the window of interest from $[r_k, r_k+D_k)$ to $[r_k-y,
r_k+D_k)$. The variable
$\omega_i$ defines the maximum residual execution time $\omega_i T_i$
of a carry-in job of task $\tau_i$ that arrives before $r_k-y$ and
should be executed in the window of interest, i.e., $[r_k-y,
r_k+D_k)$. If the residual workload is at most $\omega_i T_i$, the
next job can be released at time $r_k-y+\omega_i$, as shown in the proof.
Task $\tau_k$ is schedulable by the global fixed-priority scheduling,
if, for any combinations of $y \geq 0$ and $0 \leq \omega_i \leq T_i,
\forall \tau_i \in hp(\tau_k)$, we can always finish the inflated workload $\Delta_k^{\max} \cdot(D_k+y)$ of
task $\tau_k$ and the higher-priority workload in the window of
interest. For formal explanations, please refer to the formal proof of Lemma~\ref{lemma:multiprocessor-grm-sporadic-pushforward}.

Note that the schedulability condition in
Lemma~\ref{lemma:multiprocessor-grm-sporadic-pushforward} requires to
test all possible $y \geq 0$ and all possible settings of $0 \leq
\omega_i \leq T_i$ for the higher-priority tasks $\tau_i$ with
$i=1,2,\ldots,k-1$.  Therefore, it needs exponential time (for all the
possible combinations of $\omega_i$).\footnote{This may be the reason why the authors in \cite{DBLP:journals/rts/BaruahBMS10} did not exploit this test.} However, we are not going to directly use
the test in Lemma~\ref{lemma:multiprocessor-grm-sporadic-pushforward}
in the paper. We will only use this test to construct the
corresponding $k$-point schedulability test under
Definition~\ref{def:kpoints}.

We present the corresponding polynomial-time schedulability tests for
global fixed-priority scheduling. More specifically, we will also
analyze the capacity augmentation factors of these tests for global RM
and global DM in Corollaries~\ref{col:grm-tight}
and~\ref{col:gdm-tight}, respectively.

\begin{theorem}
  \label{thm:multiprocessor-grm-sporadic-tight}
Let $U_k^{\max}$ be $\max_{j=1}^{k} U_j$.  
Task $\tau_k$ in a sporadic task system with implicit deadlines  is schedulable by global RM on
$M$ processors if
\begin{equation}
\label{eq:schedulability-GRM-tight}
U_k^{\max} \leq 1 - \frac{2}{M}\sum_{i=1}^{k-1} U_i + \frac{0.5}{M^2}\left((\sum_{i=1}^{k-1} U_i)^2 +(\sum_{i=1}^{k-1} U_i^2)\right)
\end{equation}
or
\begin{equation}
\label{eq:schedulability-GRM-tight-ubound}
\frac{\sum_{j=1}^{k-1} U_j}{M}\leq \left(\frac{k-1}{k}\right)\left(
  2-\sqrt{2+ 2U_k^{\max}\frac{k}{k-1}}\right).
\end{equation}
\end{theorem}
\begin{proof}
  We will show that the schedulability condition in the theorem holds for all possible settings of $y$ and $\omega_i$s.
  Suppose that $y$ and $\omega_i$ for $i=1,2,\ldots,k-1$  are given,
  in which  $y \geq 0$ and $0 \leq \omega_i \leq T_i$. Let $t_k$ be
  $T_k+y$. Now, we set $t_i$ to $\omega_i +
  \floor{\frac{T_k+y-\omega_i}{T_i}}T_i$ for $i=1,2,\ldots,k-1$ and
  reindex the tasks such that $t_1 \leq t_2 \leq \ldots \leq t_k$. 

  Therefore, if $i < j$, we know that $\omega_i\cdot U_i+
  \ceiling{\frac{t_j-\omega_i }{T_i}} C_i \leq \omega_i\cdot U_i+
  (\floor{\frac{t_i-\omega_i }{T_i}}+1)C_i  = t_i U_i + C_i$. If $i
  \geq j$, then $\omega_i\cdot U_i+
  \ceiling{\frac{t_j-\omega_i }{T_i}} C_i \leq \omega_i\cdot U_i+
  (\ceiling{\frac{t_i-\omega_i }{T_i}})C_i  = t_i U_i$.
The sufficient schedulability condition in Lemma~\ref{lemma:multiprocessor-grm-sporadic-pushforward} under the given $y$ and $\omega_i$s is to verify the existence of $t_j \in
\setof{t_1, t_2, \ldots t_k}$ such that
\begin{align}
  &U_k^{\max}(T_k+y) + \frac{\sum_{i=1}^{k-1}  \omega_i\cdot U_i+ \ceiling{\frac{t_j-\omega_i }{T_i}} C_i}{M} \\ 
 \leq & U_k^{\max} t_k + \frac{\sum_{i=1}^{k-1} U_i t_i + \sum_{i=1}^{j-1} C_i}{M} 
  \leq t_j.  
\end{align}

  By the definition of global RM scheduling (i.e., $T_k \geq T_i$), we
  can conclude that $C_i = U_i T_i \leq U_i T_k \leq U_i (T_k+y) = U_i t_k$ for $i=1,2,\ldots,k-1$.
  Hence, we can safely reformulate the sufficient test by verifying whether there exists
 $t_j \in
\setof{t_1, t_2, \ldots t_k}$ such that
\begin{align}
  \label{eq:final-grm-k2u-precondition}
U_k^{\max} t_k + \sum_{i=1}^{k-1} \alpha_i U_i t_i + \sum_{i=1}^{j-1} \beta U_i t_k
  \leq t_j,
\end{align}
where $\alpha_i = \frac{1}{M}$ and $\beta \leq \frac{1}{M}$ for $i=1,2,\ldots,k-1$.
Therefore, we reach the conclusion of the schedulability conditions in
Eqs.~\eqref{eq:schedulability-GRM-tight} and
\eqref{eq:schedulability-GRM-tight-ubound} by
Lemma~\ref{lemma:framework-constrained-schedulability} and
Lemma~\ref{lemma:framework-totalU-exclusive} under given $y$ and
$\omega_i$s, respectively.

The schedulability test in Eq.~\eqref{eq:final-grm-k2u-precondition} is independent from the settings of $y$ and $\omega_i$s. However, the setting of $y$ and $\omega_i$s affects how the $k-1$ higher-priority tasks are indexed. Fortunately, it can be observed that the schedulability tests in Eqs.~\eqref{eq:schedulability-GRM-tight} and \eqref{eq:schedulability-GRM-tight-ubound} are completely independent upon the indexing of the higher-priority tasks. Therefore, no matter how $y$ and $\omega_i$s are set, the schedulability conditions in Eqs.~\eqref{eq:schedulability-GRM-tight} and \eqref{eq:schedulability-GRM-tight-ubound} are the corresponding results from the \frameworkkq{} framework.
As a result, we can reach the conclusion.
\end{proof}

\begin{corollary}
  \label{col:grm-tight}
  The capacity augmentation factor of global RM for a sporadic system
  with implicit deadlines is $\frac{3+\sqrt{7}}{2}\approx 2.823$.
\end{corollary}
\begin{proof}
  Suppose that $\sum_{\tau_i} \frac{C_i}{M T_i} \leq \frac{1}{b}$ and
  $U_k^{\max} \leq \max_{\tau_i} U_i \leq \frac{1}{b}$.  The right-hand side
  of Eq.~\eqref{eq:schedulability-GRM-tight-ubound} converges to
  $2-\sqrt{2+U_k^{\max}}$ when $k\rightarrow \infty$ Therefore,
  by Eq.~\eqref{eq:schedulability-GRM-tight-ubound}, we can guarantee
  the schedulability of task $\tau_k$ if $\frac{1}{b} \leq
  2-\sqrt{2+\frac{2}{b}}$. This is equivalent to solving $x =
  2-\sqrt{2+2x}$, which holds when $x = 3-\sqrt{7}$.
  Therefore, we reach the conclusion of the capacity augmentation
  factor $\frac{3+\sqrt{7}}{2}\approx 2.823$.
\end{proof}

\begin{theorem}
  \label{thm:multiprocessor-gdm-sporadic-tight}
Let $\Delta_k^{\max}$ be $\max_{j=1}^{k-1} \{U_j, \frac{C_k}{D_k}\}$.    
Task $\tau_k$ in a sporadic task system with constrained deadlines  is
schedulable by a global fixed-priority scheduling on
$M$ processors if   $\sum_{i=1}^{k} U_i \leq M$,
$\sum_{i=1}^{k-1}\frac{C_i}{D_k} \leq M$, and 
\begin{equation}
\label{eq:schedulability-GDM-tight}
\Delta_k^{\max} \leq 1 - \frac{1}{M}\sum_{i=1}^{k-1} \left( U_i + \frac{C_i}{D_k}\right )+ \frac{1}{M^2}\left(\sum_{i=1}^{k-1} U_i (\sum_{\ell=i}^{k-1} \frac{C_\ell}{D_k})\right),
\end{equation}
  where the $k-1$ higher-priority tasks are ordered
  in a non-increasing order of their periods.
 \end{theorem}
\begin{proof}
  This is due to a similar proof to that of
  Theorem~\ref{thm:multiprocessor-grm-sporadic-tight} and
  Lemma~\ref{lemma:multiprocessor-grm-sporadic-pushforward}, by
  applying Lemma~\ref{lemma:framework-general-schedulability} with
  $t_k=D_k+y$, $\alpha_i=\frac{1}{M}$, and $\beta_i = \frac{1}{M}$, under the worst-case last release time
  ordering, $\frac{\beta_i C_i}{\alpha_i U_i}=T_i$ non-increasingly, in Lemma~\ref{lemma:general-sorting}. Therefore, for a
  given $y \geq 0$, if $\Delta_k^{\max} \leq 1 -
  \frac{1}{M}\sum_{i=1}^{k-1} \left( U_i + \frac{C_i}{D_k+y}\right )+
  \frac{1}{M^2}\left(\sum_{i=1}^{k-1} U_i (\sum_{\ell=i}^{k-1}
    \frac{C_\ell}{D_k+y})\right)$, task $\tau_k$ is schedulable by
  the global fixed-scheduling.  By the assumption $\sum_{i=1}^{k-1} U_i \leq M$, we know that 
$ C_i \frac{\sum_{\ell=1}^{i} U_\ell}{M}  \leq C_i$.  Therefore\footnote{This comes from the simple algebra property that for any two vectors $\vec{a}$ and $\vec{b}$ of size $(k-1)$ there is $\sum_{i=1}^{k-1} a_i \sum_{\ell=i}^{k-1} b_\ell = \sum_{i=1}^{k-1} b_i \sum_{\ell=1}^{i} a_\ell$.},
  \begin{align*}
&-\frac{1}{M}\sum_{i=1}^{k-1} \frac{C_i}{D_k+y} +
\frac{1}{M^2}\left(\sum_{i=1}^{k-1} U_i (\sum_{\ell=i}^{k-1}
  \frac{C_\ell}{D_k+y})\right)\\
= &    \frac{1}{D_k+y}\frac{1}{M} \left(-\sum_{i=1}^{k-1} C_i
+\sum_{i=1}^{k-1} C_i \left(\frac{\sum_{\ell=1}^{i} U_\ell}{M} \right)\right)
  \end{align*}
  is minimized when $y$ is $0$.  As a result, the above schedulability
  condition is the worst when $y$ is $0$.
\end{proof}

\begin{corollary}
  \label{col:gdm-tight}
  The capacity augmentation factor and the speed-up factor of global
  DM by using Theorem~\ref{thm:multiprocessor-gdm-sporadic-tight} for
  a sporadic system with constrained deadlines is 3.
\end{corollary}
\begin{proof}
  If $\sum_{i=1}^{k} U_i \leq M$ or $\sum_{i=1}^{k-1}\frac{C_i}{D_k}
  \leq M$ is violated, the capacity augmentation factor is already
  $1$. Therefore, we focus on the case that task $\tau_k$ does not
  pass the schedulability condition in
  Eq.~\eqref{eq:schedulability-GDM-tight}. That is,
  \begin{align*}
\Delta_k^{\max} >& 1 - \frac{1}{M}\sum_{i=1}^{k-1} \left( U_i +
  \frac{C_i}{D_k}\right )+ \frac{1}{M^2}\left(\sum_{i=1}^{k-1} U_i
  (\sum_{\ell=i}^{k-1} \frac{C_\ell}{D_k})\right) \\
  \geq& 1 - \frac{1}{M}\sum_{i=1}^{k-1} \left( U_i +
  \frac{C_i}{D_k}\right).
  \end{align*}
  This means that the unschedulability of task $\tau_k$ under global
  DM implies that either $\Delta_k^{\max} > \frac{1}{3}$,
  $\frac{1}{M}\sum_{i=1}^{k-1}  U_i > \frac{1}{3}$,  or
  $\frac{1}{M}\sum_{i=1}^{k-1}  \frac{C_i}{D_k} > \frac{1}{3}$, by the
  pigeonhole principle. Therefore, we conclude the factor $3$.
\end{proof}

\noindent{\bf Remarks:} The utilization bound in
Eq.~\eqref{eq:schedulability-GRM-tight-ubound} is analytically better
than the best known utilization-based schedulability test
$\sum_{j=1}^{k} U_j\leq \frac{M}{2}(1-U_k^{\max})+U_k^{\max}$ for
global RM by Bertogna et al. \cite{DBLP:conf/opodis/BertognaCL05},
since $\frac{(1-x)}{2} \leq 2-\sqrt{2+2x}$ when $0 \leq x \leq 1$. 

The capacity augmentation factor $2.823$ in Corollary
\ref{col:grm-tight} is weaker than the result $2.668$ by Lundberg
\cite{DBLP:conf/rtas/Lundberg02}.  However, we would like to point out
the incompleteness in the proof in
\cite{DBLP:conf/rtas/Lundberg02}. In the proof of the extreme task
set, the argument in Page 150 in \cite{DBLP:conf/rtas/Lundberg02}
concludes that task $\tau_n$ is more difficult to be schedulable due
to the increased interference of task $\tau_{n-1}$ to task $\tau_n$
after the transformation. The argument was not correctly proved and
can be  optimistic since
the increased interference has to be analyzed in all time points in
the analysis window, whereas the analysis in
\cite{DBLP:conf/rtas/Lundberg02} only considers the interference in a
specific interval length. Without analyzing the resulting interference
in all time points in the analysis window, task $\tau_n$ after
transformation may still have chance to finish earlier due to the
potential reduction of the interference at earlier time points.

The speed-up factor $3$ provided in Corollary~\ref{col:gdm-tight} is
asymptotically the same as the result by Baruah et
al. \cite{DBLP:journals/rts/BaruahBMS10}. The speed-up factor $3$ in
\cite{DBLP:journals/rts/BaruahBMS10} requires a pseudo polynomial-time
test, whereas we show that a simple test in
Eq.~\eqref{eq:schedulability-GDM-tight} can already yield the speed-up
factor $3$ in $O(k\log k)$ time complexity.

\ifbool{techreport}{We limit our attention here for the global RM/DM scheduling. Andersson et
al. \cite{DBLP:conf/rtss/AnderssonBJ01} propose the RM-US[$\varsigma$]
algorithm, which gives the highest priority to tasks $\tau_i$s with
$U_i > \varsigma$, and otherwise assigns priorities by using RM. Our
analysis here can also be applied for the RM-US[$\varsigma$] algorithm
with some modifications in the proofs by setting
$\varsigma=\frac{2}{3+\sqrt{7}}\approx 0.3542$.}{}

\begin{figure}[t]
  \centering
  \scalebox{0.45}{
    \begin{tikzpicture}[x=1cm,font=\sffamily\Large,ultra thick]]
   
  \draw[->](0,0) -- (19, 0)  node[right]{$x$};
  
   \draw[-](2,-.1) node[below]{$z_\ell$} -- (2, 1)  -- (5,1) node[pos=.5,above]{$J_\ell$} -- (5,0) ;
     
    \draw[-](4,-.1) node[below]{$z_{\ell-1}$} -- (4, 1.2)  -- (7,1.2) node[pos=.5,above]{$J_{\ell-1}$};

\draw[-](10,1) node[above]{$J_{3}$} -- (13,1) -- (13,0) ;

 \draw[-](12,-.1) node[below]{$z_2$} -- (12, 1.2)  -- (16,1.2) node[pos=.5,above]{$J_2$} -- (16,0) ;

 \draw[-](15,-.1) node[below]{$z_1$} -- (15, 1)  -- (18,1) node[pos=.5,above]{$J_1$} -- (18,0)node[below]{$z_0$} ;
 
  \node [ ] at (8,0.6) {$\cdot\cdot\cdot\cdot\cdot$};
    \node [ ] at (1,0.6) {$\cdot\cdot\cdot\cdot\cdot$};
      \end{tikzpicture}  }
  \caption{Notation of the forced-forward algorithm for the analysis of global fixed-priority scheduling.}
  \label{fig:force-forward}
\end{figure}
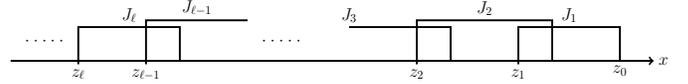

\section*{Appendix D: Response-Time for Periodic Tasks with Jitters}
\label{sec:response-jitter} 

A periodic task (with jitters) $\tau_i$ is defined by its period
$T_i$, its execution time $C_i$, its relative deadline $D_i$ and, its
jitter $L_i$. The jitter problem arises when we consider some
flexibility to delay the job arrival for a certain bounded length
$L_i$. Such a problem has also been studied in the literature, such as
\cite{DBLP:conf/rtss/BaruahCM97,238595}.  

We focus on uniprocessor fixed-priority scheduling here. The busy window concept has
also been used for the schedulability analysis of fixed-priority scheduling
\cite{238595}. 
For the $h$-th job of task $\tau_k$ in the busy window, the finishing
time $R_{k,h}$ is the minimum $t$ such that
\[ 
h C_k + \sum_{i=1}^{k-1} \ceiling{\frac{t+L_i}{T_i}}C_i \leq t.
\] 
The $h$-th job of task $\tau_k$ arrives at time $\max\{(h-1)T_k-L_k,
0\}$, and, hence, its response time is $R_{k,h}-\max\{(h-1)T_k-L_k,
0\}$. The busy window of task $\tau_k$ finishes on the $h$-th job if
$R_{k,h} \leq \max\{h T_k-L_k, 0\}$.

\begin{lemma}
\label{lemma:rkh-last-release}
$R_{k,h}$ is upper bounded by finding the maximum
  \begin{equation}
    \label{eq:PJ-arbitrary-objective-k}
   t_k = hC_k + \sum_{i=1}^{k-1} L_i U_i +\sum_{i=1}^{k-1}  t_i U_i + \sum_{i=1}^{k-1}  C_i,
  \end{equation}
  among all possible least release time orderings of the $k-1$
  higher-priority tasks with $0 \leq t_1
  \leq t_2 \leq \cdots \leq t_{k-1} \leq t_{k}$ and
  {\small \begin{align}
    \label{eq:PJ-arbitrary-objective-k2}
    hC_k + \sum_{i=1}^{k-1} L_i U_i  + \sum_{i=1}^{k-1} t_i U_i + \sum_{i=1}^{j-1} C_i > t_j, & \forall j=1,2,\ldots,k-1.
  \end{align}  }
\end{lemma}
\begin{proof}
  Suppose that $R_{k,h}$ is known. Then, the last release of task
  $\tau_i$ in the busy window before $R_{k,h}$ is at time
  $t_i=(\ceiling{\frac{R_{k,h}+L_i}{T_i}}-1)T_i-L_i$.  We can index the
  $k-1$ higher-priority tasks by $t_i$ non-decreasingly. 
  By the definition of $t_i$, we know that
  $(\ceiling{\frac{R_{k,h}+L_i}{T_i}}-1) C_i =
  \frac{t_i+L_i}{T_i} T_iU_i = t_i U_i +
  L_i U_i$. That is, the accumulative workload of task $\tau_i$ from
  $0$ to $t_i$ is exactly $L_iU_i + t_i U_i$. 
 Due to the indexing
  rule, and the known $R_{k,h}$ we can now conclude the conditions in
  Eqs.~\eqref{eq:PJ-arbitrary-objective-k} and~\eqref{eq:PJ-arbitrary-objective-k2}.

  However, it should be noted that the last release time ordering
  is actually unknown since $R_{k,h}$ is unknown. Therefore, we
  have to consider all possible last release time orderings.
\end{proof}

\begin{lemma}
  \label{lemma:finishing-time-pj-h}
  Suppose that $\sum_{i=1}^{k-1}  U_i \leq 1$. Then, for any
  $h \geq 1$ and $C_k > 0$, we have
  \begin{equation}
    \label{eq:PJ-R-k-h}
     R_{k,h} \leq \frac{hC_k+\sum_{i=1}^{k-1}  (C_i + L_iU_i) - \sum_{i=1}^{k-1} U_i (\sum_{\ell=i}^{k-1}  C_{\ell})}{1-\sum_{i=1}^{k-1}  U_i},
  \end{equation}
  where the $k-1$ higher-priority tasks are ordered
  in a non-increasing order of their periods.
\end{lemma}
\begin{proof}
 This comes from the above discussions with $\alpha_i=1$, $\beta_i=1$
 in Lemma~\ref{lemma:rkh-last-release} 
  by applying Lemmas~\ref{lemma:framework-general-response} and
  \ref{lemma:general-response-sorting} when $\sum_{i=1}^{k-1} 
  U_i < 1$. The case when $\sum_{i=1}^{k-1} U_i = 1$ has a
  safe upper bound $R_{k,h}=\infty$ in Eq. \eqref{eq:PJ-R-k-h}.
\end{proof}

\begin{theorem}
  \label{theorem:response-time-PJ}
  Suppose that $\sum_{i=1}^{k} U_i \leq 1$. The worst-case
  response time of task $\tau_k$ is at most
{\small
  \begin{align}
    \label{eq:jitter-R-k}
    R_{k} \leq & \frac{\sum_{i=1}^{k-1}    (C_i + L_iU_i)-
      \sum_{i=1}^{k-1} U_i (\sum_{\ell=i}^{k-1}
      C_{\ell})}{1-\sum_{i=1}^{k-1}  U_i},\\
&               + \max\left\{\frac{h^*C_k}{1-\sum_{i=1}^{k-1}  U_i},
  \frac{(h^*+1)C_k}{1-\sum_{i=1}^{k-1}  U_i} - h^*T_k+L_k\right\} \nonumber
  \end{align}
}
  where the $k-1$ higher-priority tasks are ordered
  in a non-increasing order of their periods and $h^*=\floor{\frac{L_k}{T_k}}+1$.
\end{theorem}
\begin{proof}
By the definition of $h^*$,
  we know that $(h^*-1)T_k-L_k \leq 0$, whereas $h^*T_k-L_k > 0$.
  When $h \leq h^*$, we know that $\max\left\{(h^*-1)T_k-L_k, 0\right\}$ is
$0$. Therefore, $R_{k,h}-0$ is maximized when $h$ is set to $h^*$ for
any $h \leq h^*$.

The first-order derivative of $R_{k,h}-((h-1)T_k - L_k)$ with respect to $h$
when $h \geq h^*+1$ is $\frac{C_k}{1-\sum_{i=1}^{k-1} U_i} - T_k =
\frac{C_k - (1-\sum_{i=1}^{k-1} U_i) T_k }{1-\sum_{i=1}^{k-1}
  U_i}$. Similar to the proof of
Theorem~\ref{theorem:response-time-sporadic}, we know that setting
$R_{k,h}-(h-1)T_k$ is maximized when $h$ is set to $h^*+1$ for any $h
\geq h^*+1$.

  As a result, we only have to evaluate the two cases by setting $h$
  as $h^*$ or $h^*+1$. One of them is the worst-case response
  time. The formulation in Eq.~\eqref{eq:jitter-R-k} simply compares
  these two response times.
\end{proof}

\section*{Appendix E: Schedulability for Generalized Multi-Frame}
\label{sec:schedulability-gmf}

A generalized multi-frame real-time task $\tau_i$ with $m_i$ frames is
defined as a task with an array $(C_{i,0}, D_{i,0}, T_{i,0}, C_{i,1},
D_{i,1}, T_{i,1},
\ldots, C_{i,m_i-1}, D_{i,m_i-1}, T_{i,m_i-1})$ of different execution times,
different relative deadlines, and
the minimal inter-arrival time of the next frame
\cite{DBLP:conf/rtcsa/TakadaS97,DBLP:journals/rts/BaruahCGM99}. The execution time
of the $j$-th job of task $\tau_i$ is defined as $C_{i,(j\mod
  m_i)}$. If a job of the $j$-th frame of task $\tau_i$ is released at
time $t$, the earliest time that task $\tau_i$ can release the next
$(j+1) \mod m_i$ frame is at time $t+T_j$. Here, we consider only
constrained-deadline cases, in which $D_{i,j} \leq T_{i,j}$.

Takada and Sakamura \cite{DBLP:conf/rtcsa/TakadaS97} provide an exact
test with exponential-time complexity for such a problem under
task-level fixed-priority scheduling, in which each task is assigned
with one static priority level. For this section, we will focus on
such a setting. Specifically, we are interested in analyzing the
schedulability of the $h$-th frame of task $\tau_k$ under the given
task priority ordering. It was shown that the critical instant theorem
of periodic task systems by Liu and Layland \cite{liu1973scheduling}
does not work anymore for generalized multi-frame
systems. Fortunately, as shown in Theorem 2 in
\cite{DBLP:conf/rtcsa/TakadaS97}, the critical instant of the $h$-th
frame of task $\tau_k$ is to release a certain frame of a higher
priority task $\tau_i$ at the same time, and the subsequent frames of
task $\tau_i$ as early as possible. 
In fact, this problem has been
recently proved to be co-NP hard in the strong sense in \cite{DBLP:conf/ecrts/StiggeY12}.

For completeness, the test with exponential time complexity by Takada
and Sakamura \cite{DBLP:conf/rtcsa/TakadaS97} will be presented in
Lemma~\ref{lemma:generalized-mf} by using our notation. 
We define ${\sc rbf}(i, q, t)$ as the maximum workload of task
$\tau_i$ released within an interval length $t$ starting from the $q$-th 
frame. That is, 
\begin{equation}
  \label{eq:rbf-gmf}
  {\sc rbf}(i, q, t) = \sum_{j=q}^{\theta(i,q,t)} C_{i,(j\mod m_i)},
\end{equation}
where $\theta(i,q,t)$ is the smallest $\ell$ such that
$\sum_{j=q}^{\ell} T_{i,(j\mod m_i)} \geq t$. That is, $\theta(i,q,t)$ is
the last frame released by task $\tau_i$ before $t$ under the critical
instant starting with the $q$-th frame of task $\tau_i$.

\begin{lemma}
\label{lemma:generalized-mf}
The $h$-th frame of task $\tau_k$ in a generalized multiframe task
system with constrained deadlines  is
schedulable by a fixed-priority scheduling on a uniprocessor system
if 
\begin{align*}
&\forall q_i \in\setof{0, 1, \ldots, m_i-1}, \forall i=1,2,\ldots,k-1\\
&\exists 0 < t \leq D_{k,h}, \;\;\;\;C_{k,h} + \sum_{i=1}^{k-1}{\sc rbf}(i, q_i, t) \leq t.
\end{align*}
\end{lemma}
\begin{proof}
  This is a reformulation of the test by Takada and Sakamura \cite{DBLP:conf/rtcsa/TakadaS97}.
\end{proof}

Since the test in Lemma~\ref{lemma:generalized-mf} requires
exponential-time complexity, an approximation by using Maximum
Interference Function (MIF) was proposed in
\cite{DBLP:conf/rtcsa/TakadaS97} to provide a sufficient test
efficiently.\footnote{The reason why this is not an exact test was
  recently provided by Stigge and Wang in
  \cite{DBLP:conf/ecrts/StiggeY12}.}  Instead of testing all possible
combinations of $q_i$, a simple strategy is to use $rbf(i,t) =
\max_{q=0}^{m_i-1}rbf(i,q,t)$ to approximate the test in
Lemma~\ref{lemma:generalized-mf}. This results in a
pseudo-polynomial-time test. Guan et
al. \cite{DBLP:conf/rtss/GuanGSD014} have recently provided proofs to
show that such an approximation is with a  speed-up factor of $2$.

We will use a different way to build our analysis by constructing the
$k$-point last-release schedulability test in
Definition~\ref{def:kpoints}. The idea is very simple. If task
$\tau_i$ starts with its $q$-th frame, we can find and define 
\[ t_{i,q}
= \sum_{j=q}^{\theta(i,q,D_{k,h})-1} T_{i,(j\mod m_i)}
\] as the last
release time of task $\tau_i$ before $D_{k,h}$. Therefore, for the given
$D_{k,h}$, we are only interested in these $m_i$ last release times of
task $\tau_i$. We need a safe function for estimating their
workload. More precisely, we want to find two constants $U_{i,k,h}$
and $C_{i,k,h}$ such that $U_{i,k,h}\cdot t_{i,q} \geq
rbf(i,q,t_{i,q})$ and $U_{i,k,h}\cdot t_{i,q} + C_{i,k,h}\geq
rbf(i,q,t_{i,q}+\epsilon) = rbf(i,q,D_{k,h})$. This means that no matter
which frame of task $\tau_i$ is the first frame in the critical
instant, we can always bound the workload from task $\tau_i$ by using 
$U_{i,k,h}\cdot t + C_{i,k,h}$ for the points that we are interested
to test. According to the above discussions, we can set
\begin{align}
&U_{i,k,h} = \max_{q=0,1,\ldots,m_i-1}   \max\left\{\frac{rbf(i,q,t_{i,q})}{t_{i,q}}\right\} \label{U:ikh}\\
&C_{i,k,h} = \max_{q=0,1,\ldots,m_i-1}   {\sc rbf}(i, q,
D_{k,h})-U_{i,k,h}\cdot t_{i,q}\label{C:ikh}
\end{align}
Now, we can reorganize the schedulability test to link to
Definition~\ref{def:kpoints}.

\begin{lemma}
\label{lemma:generalized-mf-approximate}
The $h$-th frame of task $\tau_k$ in a generalized multiframe task
system with constrained deadlines is schedulable by a fixed-priority
scheduling on a uniprocessor system if the following condition holds:
For any last release ordering $\pi$ of the $k-1$ higher-priority
tasks with $t_{1,q_1} \leq t_{2,q_2} \leq \cdots \leq t_{k-1,q_{k-1}}
\leq t_k = D_{k,h}$, there exists $t_{j,q_j}$ such that
\begin{equation}
  \label{eq:genralized-mf-schedulability-approximate}
  C_{k,h} + \sum_{i=1}^{k-1}  U_{i,k,h}\cdot t_{i,q_i} + \sum_{i=1}^{j-1} C_{i,k,h} \leq t_{j,q_j},
\end{equation}
where $U_{i,k,h}$ and $C_{i,k,h}$ are defined in Eqs.~\eqref{U:ikh}
and~\eqref{C:ikh}, respectively.
\end{lemma}
\begin{proof}
  This comes from the above discussions and Lemma~\ref{lemma:generalized-mf}.
\end{proof}

The choice of $q_i$ only affects the last release ordering of the
$k-1$ higher-priority tasks, but it does not change the constants
$U_{i,k,h}$ and $C_{i,k,h}$. Therefore,
Lemma~\ref{lemma:generalized-mf-approximate} satisfies
Definition~\ref{def:kpoints} by removing the indexes $q$, but the
worst-case last release time ordering has to be considered.

\begin{theorem}
  \label{theorem:response-time-sporadic-gmf}
  Suppose that $\sum_{i=1}^{k-1} C_{i,k,h} < D_{k,h}$.  The $h$-th
  frame of task $\tau_k$ in a generalized multiframe task system with
  constrained deadlines is schedulable by a scheduling algorithm on a
  uniprocessor system if
  the following condition holds
\begin{equation}
\label{eq:schedulability-gmf}
\frac{C_{k,h}}{D_{k,h}} \leq 1 - \sum_{i=1}^{k-1} U_{i,k,h} - \frac{\sum_{i=1}^{k-1} (C_{i,k,h} - U_{i,k,h} (\sum_{\ell=i}^{k-1} C_{\ell,k,h}) )}{D_{k,h}},
\end{equation}
  in which the $k-1$ higher-priority tasks are indexed 
in a non-increasing order of $\frac{C_{i,k,h}}{U_{i,k,h}}$ of task $\tau_i$.
\end{theorem}
\begin{proof}
  The schedulability test (under a specified last release time
  ordering $\pi$) in Lemma \ref{lemma:generalized-mf-approximate}
  satisfies Definition~\ref{def:kpoints} with
  $\alpha_i=\beta_i=1$. Therefore, we can apply
  Lemmas~\ref{lemma:framework-general-schedulability} and
  \ref{lemma:general-sorting} to reach the conclusion of the proof.
\end{proof}

A na\"ive implementation to calculate $t_{i,q}$ and $rbf(i,q,t)$
requires $O(m_i)$ time complexity. Therefore, calculating $U_{i,k,h}$
and $C_{i,k,h}$ requires $O(m_i^2)$ time complexity with such an
implementation. It can be calculated with better data structures and
implementations to achieve $O(m_i)$ time complexity. That is, we
first calculate $t_{i,0}$ and $rbf(i, 0, t_{i,0})$ in $O(m_i)$ time
complexity. The overall time complexity to calculate the following
$t_{i,q}$ and $rbf(i,q,t_{i,q})$ can be done in $O(m_i)$ by using
simple algebra.  Therefore, the time complexity of the schedulability
test in Theorem~\ref{theorem:response-time-sporadic-gmf} is
$O(\sum_{i=1}^{k-1}m_i + k \log k)$.

\noindent{\bf Remarks:} Although we focus ourselves on generalized
multi-frame task systems in this section, it can be easily seen that
our approach can also be adopted to find polynomial-time tests based
on the request bound function presented in
\cite{DBLP:conf/rtss/GuanGSD014}. Since such functions have been shown
useful for schedulability tests in more generalized task models,
including the digraph task model \cite{DBLP:conf/rtas/StiggeEGY11},
the recurring real-time task model \cite{DBLP:conf/rtss/Baruah10},
etc., the above approach can also be applied for such models by
quantifying the utilization (e.g., $U_{i,k,h}$) and the execution time
(e.g., $C_{i,k,h}$) of task $\tau_i$
based on the \emph{limited} options of the last release times before
the deadline (e.g., $D_{k,h}$ in generalized multi-frame) to be tested.

\section*{Appendix F: Acyclic and Multi-Mode Tasks}
\label{sec:schedulability-acyclic-modes}

This section considers the acyclic model, proposed in
\cite{DBLP:journals/tc/AbdelzaherSL04}. That is, each task $\tau_i$ is
specified only by its utilization $U_i$. An instance of task $\tau_i$
can have different worst-case execution times and different relative
deadlines. If an instance of task $\tau_i$ arrives at time $t$ with
execution time $C_{i,t}$, its relative deadline is
$\frac{C_{i,t}}{U_i}$, and the next instance of task $\tau_i$ can only
be released after $t+\frac{C_{i,t}}{U_i}$. 


For systems with known modes, we can also define a multi-mode task
system. A multi-mode task $\tau_i$ with $m_i$ modes is denoted by a
set of triplet: $\tau_i=\{\tau_{i,0}=\setof{C_{i,0}, D_{i,0},
  T_{i,0}}, \tau_{i,1}=\setof{C_{i,1}, D_{i,1}, T_{i,1}}, \ldots,
\tau_{i,m_i-1}=\setof{C_{i,m_i-1}, D_{i,m_i-1}, T_{i,m_i-1}}\}$ to
specify the worst-case execution time $C_{i,j}$, the minimum
inter-arrival time $T_{i,j}$, and the relative deadline $D_{i,j}$ of
the corresponding task mode $\tau_{i,j}$.  For a multi-mode task $\tau_i$, when a
job of mode $\tau_{i,j}$ is released at time $t$, this job has to be
finished no later than its absolute deadline at time $t+D_{i, j}$, and
the next release time of task $\tau_i$ is no earlier than
$t+T_{i,j}$. We only consider systems with constrained deadlines, in
which $D_{i,j} \leq T_{i,j}$. This model is studied in our paper \cite{HuangC-RTCSA15}.
 The difference between this model and the
 generalized multi-frame model is that the system can switch arbitrarily
 among any two modes if the temporal seperation constraints
 in the multi-mode model are respected.

We will focus on mode-level fixed-priority scheduling on uniprocessor
scheduling in this section. Suppose that we are testing whether task
$\tau_{k,h}$ can be feasibly scheduled. Let $hp(\tau_{k,h})$ be the
set of task mode $\tau_{k,h}$ and the other task modes with higher priority than task mode
$\tau_{k,h}$. It is important to note that $\tau_{k,h}$ is also in $hp(\tau_{k,h})$ for the simplicity of presentation. For notational brevity, we assume that there are $k-1$
tasks with higher-priority task modes than $\tau_{k,h}$.  Moreover,
for the rest of this section, we implicitly assume that the tasks in
$hp(\tau_{k,h})\setminus\setof{\tau_{k,h}}$ are schedulable by the mode-level fixed-priority
scheduling algorithm under testing.

Let $load(i, t)$ be the maximum workload of task $\tau_i$ (by
considering all the task modes $\tau_{i,j} \in hp(\tau_{k,h})$)
released from $0$ to $t$ such that the next mode can be released at
time $t$. That is, let $n_{i,q}$ be a non-negative integer to denote
the number of times that task mode $\tau_{i,q}$ is released, in which
\begin{align*}
load(i, t) = \max_{\tau_{i,q} \in hp(\tau_{k,h})}\{&\sum_{q}
n_{i,q} C_{i,q}\\
&\qquad
\mbox{ s.t. } \sum_{q} n_{i,q} T_{i,q} \leq t\}. 
\end{align*}

We denote $C_i^{\max}(\tau_{k,h})$ the maximum mode execution time
among the higher-priority modes of task $\tau_i$ than task mode
$\tau_{k,h}$, i.e., $\max_{\tau_{i,j} \in hp(\tau_{k,h})} C_{i,j}$ for
a given index $i$.  Similarly, we denote $U_i^{\max}(\tau_{k,h})$ the
maximum mode utilization among the higher-priority modes of task
$\tau_i$ than task mode $\tau_{k,h}$, i.e., $\max_{\tau_{i,j} \in
  hp(\tau_{k,h})} \frac{C_{i,j}}{T_{i,j}}$ for a given index $i$. For
notational brevity, since we define $C_{i}^{\max}(\tau_{k,h})$ and
$U_i^{\max}(\tau_{k,h})$ by referring to $\tau_{k,h}$, we will
simplify the notation by using $C_{i}^{\max}$ and $U_i^{\max}$,
respectively, for the rest of this section.

Suppose that the last release of task $\tau_i$ in the window of
interest is at $t_i$. We define the request bound function for
such a case as follows:
\begin{align}\label{eq:rbf-mode-change}
rbf(i, t_i, t) = 
\begin{cases}
  load(i, t_i) & \mbox{ if } t \leq t_i\\
  load(i, t_i) + C_i^{\max} & \mbox{ otherwise}.
\end{cases}
\end{align}

\begin{lemma}
\label{lemma:multi-mode-pushforward} 
Task $\tau_{k,h}$ in a multi-mode task system with constrained
deadlines is schedulable by a mode-level fixed-priority scheduling
algorithm on uniprocessor systems if {\small \begin{align*} &\forall y
    \geq 0, {\;\; \forall 0 \leq t_i < y+D_{k,h}
      \forall i=1,2,\ldots,k-1\;\;} \\
    &\exists t,\; 0 < t \leq y+D_{k,h},\qquad C_{k,h} + load(k, y) +
    \sum_{i=1}^{k-1} rbf(i, t_i, t) \leq t.
\end{align*}}
\end{lemma}
\begin{proof}
 This is proved by contrapositive. Suppose that $\tau_{k,h}$ misses
  its deadline firstly at time $d_{k,h}$. We know that this job of $\tau_{k,h}$ arrives to
  the system at time $a_1=d_{k,h}-D_{k,h}$. Due to the scheduling
  policy, we know that the processor is busy executing jobs in
  $hp(\tau_{k,h})$ (recall that $\tau_{k,h}$ is also in $hp(\tau_{k,h})$) from $a_1$ to $d_{k,h}$. Let $a_0$
  be the last time point in the
  above schedule before $a_1$ such that the processor is idle or
  executing any job with lower priority than $\tau_{k,h}$. It is clear
  that $a_0$ is well-defined. Therefore, the processor executes only
  jobs in $hp(\tau_{k,h})$ from $a_0$ to $d_{k,h}$. We denote $a_1-a_0$ as
  $y$. Let the last release of the task modes of task $\tau_i$ before
  $d_{k,h}$ be $a_0+t_i$. If task $\tau_i$ does not have any release,
  we set $t_i$ to $0$. Therefore, we have $0 \leq t_i < y + D_{k,h}$.

  Without loss of generality, we set $a_0$ to $0$.  Therefore, the
  workload requested by the modes of task $\tau_i$ at any time $t$
  before (and at) $t_i$ is no more than $load(i,t)$. Moreover, the
  workload released by the modes of task $\tau_i$ at any time $t$
  after $t_i$ is no more than $load(i,t_i)+C_i^{\max}$.  As a result,
  the function $rbf(i, t_i, t)$ defined in
  Eq.~\eqref{eq:rbf-mode-change} is a safe upper bound of the workload
  requested by task $\tau_i$ upt to time $t$.

  By the
  definition of the task model, task mode $\tau_{k,h}$ and the other higher-priority task modes
  of task $\tau_k$ may also release some workload before $a_1$, in
  which the workload is upper bounded by $load(k, y)$.
  The assumption of non-schedulability of task $\tau_{k,h}$ in the
  above schedule and the busy execution in the
  interval $(a_0, d_{k,h}]$ implies that the above workload at any
  point $t$ in the interval $(a_0, d_{k,h}]$ is larger than
  $t$. Therefore, we know that 
  \begin{align*}
&\exists y
    \geq 0, {\;\; \exists 0 \leq t_i < y+D_{k,h}
      \forall i=1,2,\ldots,k-1\;\;} \\    
&\forall t,\; 0 < t \leq y+D_{k,h},\;\; C_{k,h} + load(k, y) +
    \sum_{i=1}^{k-1} rbf(i, t_i, t) > t,
  \end{align*}
  which concludes the proof by contrapositive.
\end{proof}

Evaluating $load(i, t)$ is in fact equivalent to the \emph{unbounded
  knapsack problem} (UKP). The UKP problem is to select some items in a
  collection of items with specified weights and profits so that the total weight of the selected items is
  less than or equal to a given limit (called knapsack) and total
  profit of the selected items is maximized, in which an item can be
  selected unbounded multiple times. The definition of $load(i,t)$ is
  essentially an unbounded knapsack problem, by considering $t$ as the
  knapsack constraint, the minimum inter-arrival time (period) of a task mode as
  the weight of an item, and the execution time of a task mode as the
  profit of the item.
\begin{lemma}
  \label{lemma:load-mode-change}
  \begin{equation}
    \label{eq:load-mode-change}
    load(i, t) \leq U_i^{\max}\cdot t.
  \end{equation}
\end{lemma}
\begin{proof}
  It has been shown
  in~\cite{martello1990knapsack} that an upper bound of the above unbounded
  knapsack problem is $\left(\max_{\tau_{i,j} \in
    hp(\tau_{k,h})}\{\frac{C_{i,j}}{T_{i,j}}\}\right)\cdot t =
U_i^{\max}\cdot t$, which completes the proof. 
\end{proof}

By Lemma~\ref{lemma:multi-mode-pushforward} and
Lemma~\ref{lemma:load-mode-change}, we can now apply the \frameworkkq{}
framework to obtain the polynomial-time schedulability test in the
following theorems.
\begin{theorem}
  \label{theorem:schedulability-mode-change}
  For a given task $\tau_i$ and a task mode $\tau_{k,h}$ under
  testing, let $C_i^{\max} = \max_{\tau_{i,j} \in hp(\tau_{k,h})}
  C_{i,j}$ and $U_i^{\max} = \max_{\tau_{i,j} \in hp(\tau_{k,h})}
  \frac{C_{i,j}}{T_{i,j}}$.  Suppose that there are $k-1$ tasks with
  higher-priority modes than task mode $\tau_{k,h}$.  Let
  $\Delta_k^{\max}$ be $\max\{\max_{\tau_{k,j} \in hp(\tau_{k,h})}\{
  \frac{C_{k,j}}{T_{k,j}}\}, \frac{C_{k,h}}{D_{k,h}}\}$.  Suppose that
  $\sum_{i=1}^{k-1} C_i^{\max} < D_{k,h}$.  Task $\tau_{k,h}$ in a
  multi-mode task system with constrained deadlines is schedulable by
  a mode-level fixed-priority scheduling algorithm on a uniprocessor
  system if $\sum_{i=1}^{k-1} U_i^{\max} \leq 1$ and
\begin{equation}
\label{eq:schedulability-mode-change}
\Delta_k^{\max}
\leq 1 - \sum_{i=1}^{k-1} U_i^{\max} -
\frac{\sum_{i=1}^{k-1} (C_i^{\max} - U_i^{\max} (\sum_{\ell=i}^{k-1} C_{\ell}^{\max}) )}{D_{k,h}},
\end{equation}
  in which the $k-1$ higher-priority tasks are indexed 
in a non-increasing order of $\frac{C_i^{\max}}{U_i^{\max}}$ of task $\tau_i$.
\end{theorem}
\begin{proof}
  Suppose that $t_i$ is given for each task $\tau_i$. For the given
  $t_i$, we can define the last release time ordering, in which $t_1
  \leq t_2 \leq \ldots \leq t_{k-1} \leq D_{k,h}+y = t_k$. By Lemma
  \ref{lemma:load-mode-change}, we can pessimistically rephrase the
  schedulability test in Lemma~\ref{lemma:multi-mode-pushforward} to
  verify whether there exists $t_j \in \setof{t_1, t_2, \ldots, t_k}$
  such that $\Delta_{k}^{\max}(D_{k,h}+y) + \sum_{i=1}^{k-1} U_i^{\max}t_i +
  \sum_{i=1}^{j-1} C_i^{\max}\leq t_j$. Therefore, we reach
  Definition~\ref{def:kpoints} with $\alpha_i=1$ and $\beta_i=1$, and,
  hence, can apply Lemmas~\ref{lemma:framework-general-schedulability}
  and \ref{lemma:general-sorting}. .
  Moreover, with the same argument in the proof of
  Theorem~\ref{thm:multiprocessor-gdm-sporadic-tight}, we know that
  the resulting schedulability condition is the worst if $y$ is $0$
  under the assumption  $\sum_{i=1}^{k-1} U_i^{\max} \leq 1$.
\end{proof}

\begin{theorem}
  \label{theorem:schedulability-mode-change-rm}
  For a given task $\tau_i$ and a task mode $\tau_{k,h}$ under
  testing, let $U_i^{\max} = \max_{\tau_{i,j} \in hp(\tau_{k,h})}
  \frac{C_{i,j}}{T_{i,j}}$. Suppose that there are $k-1$ tasks with
  higher-priority modes than task mode $\tau_{k,h}$. Let $U_k^{\max}$
  be $\max\{\max_{\tau_{k,j} \in hp(\tau_{k,h})}\{
  \frac{C_{k,j}}{T_{k,j}}\}, \frac{C_{k,h}}{T_{k,h}}\}$.  
  Suppose that $\sum_{i=1}^{k} U_i^{\max}\leq 1$.
 Task
  $\tau_{k,h}$ in a multi-mode task system with implicit deadlines is
  schedulable by the mode-level RM scheduling algorithm on a
  uniprocessor system if the following condition holds
\begin{equation}
\label{eq:schedulability-mode-change-rm}
U_k^{\max} \leq 1 - 2\sum_{i=1}^{k-1} U_i^{\max} + 0.5\left((\sum_{i=1}^{k-1} U_i^{\max})^2 +(\sum_{i=1}^{k-1} (U_i^{\max})^2)\right),
\end{equation}
or
\begin{equation}
\label{eq:schedulability-mode-change-rm-ubound}
\sum_{i=1}^{k-1} U_i^{\max}\leq \left(\frac{k-1}{k}\right)\left(
  2-\sqrt{2+ 2U_k^{\max}\frac{k}{k-1}}\right),
\end{equation}
or
\begin{equation}
\label{eq:schedulability-mode-change-rm-ubound-2}
U_{k}^{\max}+\sum_{i=1}^{k-1} U_i^{\max}\leq 
\begin{cases}
\left(\frac{k-1}{k}\right)\left(
  2-\sqrt{4- \frac{2k}{k-1}}\right),&\mbox{ if } k > 3\\
1 - \frac{(k-1) }{2k}   &\mbox{otherwise}.
\end{cases}
\end{equation}
\end{theorem}
\begin{proof}
  Due to the RM property, we know that $T_{i,j} \leq T_{k,h}$ if
  $\tau_{i,j}$ is in $hp(\tau_{k,h})$. Therefore, $C_i^{\max} \leq
  T_{k,h} U_i^{\max}$. We now reach the conditions in
  Lemmas~\ref{lemma:framework-constrained-schedulability},
  ~\ref{lemma:framework-totalU-exclusive},
  and~\ref{lemma:framework-totalU-constrained} with
  $\alpha =1$ and $\beta =1$. The three conditions are directly from
  these lemmas.
\end{proof}

Note that the above tests assume implicitly that the tasks in
$hp(\tau_{k,h})$ are already tested to be schedulable under the
scheduling policy. Therefore, we
have to apply the results in
Theorems~\ref{theorem:schedulability-mode-change}
and~\ref{theorem:schedulability-mode-change-rm} by testing all the
task modes from the highest priority to the lowest priority. The
following theorem provides the utilization bound for testing the whole
task set in linear time for mode-level
RM scheduling. 

\begin{corollary}
  \label{col:acyclic}
  Suppose that $\sum_{i=1}^{k} U_i\leq 1$.
  A system with $k$ acyclic tasks with utilization $U_1, U_2, \ldots, U_k$ is
  schedulable by the mode-level RM scheduling algorithm on a
  uniprocessor system if 
\begin{equation}
\label{eq:schedulability-acyclicrm-ubound}
\sum_{i=1}^{k} U_i\leq
\begin{cases}
\left(\frac{k-1}{k}\right)\left(
  2-\sqrt{4- \frac{2k}{k-1}}\right),&\mbox{ if } k > 3\\
1 - \frac{(k-1) }{2k}   &\mbox{otherwise}.
\end{cases}
\end{equation}
or 
\begin{equation}
\label{eq:schedulability-acyclic-rm}
0 \leq 1 - U_k - 2\sum_{i=1}^{k-1} U_i + 0.5\left((\sum_{i=1}^{k-1} U_i)^2 +(\sum_{i=1}^{k-1} U_i^2)\right),
\end{equation}
where task $\tau_k$ is the task with the minimum utilization among the
$k$ tasks.
\end{corollary}
\begin{proof}
  This comes with similar arguments in
  Theorem~\ref{theorem:schedulability-mode-change-rm}. Adopting the
  utilization bound in Eq.~\eqref{eq:schedulability-acyclicrm-ubound}
  does not need to consider the ordering of the tasks. However, the
  quadratic bound in Eq.~\eqref{eq:schedulability-acyclic-rm} changes
  for different settings of $U_k$. Let $\sum_{i=1}^{k} U_i$ be a given
  constant $H_1$ and $\sum_{i=1}^{k} U_i^2$ be $H_2$. Then, the
  quadratic bound in Eq.~\eqref{eq:schedulability-acyclic-rm} becomes
  $1-2H_1+U_k+0.5H_1^2+0.5H_2-0.5U_k^2$.  The first order derivative
  of the quadratic bound with respect to $U_k$ is $1-U_k$, which
  implies that the minimum $U_k$ leads to the worst condition in the
  quadratic bound in Eq.~\eqref{eq:schedulability-acyclic-rm} under
  the condition $U_k \leq 1$.
\end{proof}

The result in Eq.~\eqref{eq:schedulability-acyclicrm-ubound} in
Corollary~\ref{col:acyclic} is the same as the utilization bound
$2-\sqrt{2}$ (when $k \rightarrow \infty$) in
\cite{DBLP:journals/tc/AbdelzaherSL04,HuangC-RTCSA15} for such a task model. Our
results here are more generic than \cite{DBLP:journals/tc/AbdelzaherSL04} and can also be easily applied for any
(mode-level or task-level) fixed-priority scheduling.

\end{document}